\documentclass[12pt]{article}
\pdfoutput=1
\usepackage[square,comma,numbers,sort&compress]{natbib}
\usepackage{graphicx,epstopdf,amssymb,amsfonts,amsmath,amsthm,array,
mathrsfs,amscd,enumitem}
\usepackage{float}
\usepackage{xcolor}
\usepackage{appendix}
\usepackage{fnpct}
\usepackage{hyperref}
\usepackage{lmodern}
\usepackage[T1]{fontenc}
\usepackage[utf8]{inputenc}

\newcommand{\pbs}[1]{\let\temp=\\#1\let\\=\temp}
\numberwithin{equation}{section}
\def\be{\begin{equation}}\def\ee{\end{equation}}

\def\cvp{\raise 2pt\hbox{,}}

 \def\d{{\rm d}} 
\def\la{\lambda}

\def\arctanh{\mathop{\text{arctanh}}\nolimits}

\def\XXint#1#2#3{{\setbox0=\hbox{$#1{#2#3}{\int}$}
     \vcenter{\hbox{$#2#3$}}\kern-.515\wd0}}

\def\XXoint#1#2#3{{\setbox0=\hbox{$#1{#2#3}{\int}$}
     \vcenter{\hbox{$#2#3$}}\kern-.515\wd0}}

\DeclareSymbolFont{matha}{OML}{txmi}{m}{it}
\DeclareMathSymbol{\varv}{\mathord}{matha}{118}

\def\disk{\mathscr D}

\newtheorem{lemma}{Lemma}[section]
\newtheorem{proposition}{Proposition}[section]

\def\npb#1#2#3{{\it Nucl.\ Phys.\ }{\bf B #1} (#2) #3}

\def\jhep#1#2#3{{\it J. High Energy Phys.\ }{\bf #1} (#2) #3}
\def\prd#1#2#3{{\it Phys.\ Rev.\ }{\bf D #1} (#2) #3}

\def\cmp#1#2#3{{\it Comm.\ Math.\ Phys.\ }{\bf #1} (#2) #3}

\def\imath#1#2#3{{\it Invent math }{\bf #1} (#2) #3}

\begin{document}
{\pagestyle{empty}
\parskip 0in
\

\vfill
\begin{center}

{\LARGE Neumann scalar determinants \\
\vspace{0.3cm}on constant curvature disks}

\vspace{0.4in}

Soumyadeep \textsc{Chaudhuri} 

\medskip
{\it Service de Physique Th\'eorique et Math\'ematique\\
Universit\'e Libre de Bruxelles (ULB) and International Solvay Institutes\\
Campus de la Plaine, CP 231, B-1050 Bruxelles, Belgique}

\smallskip
{\tt chaudhurisoumyadeep@gmail.com}
\end{center}
\vfill\noindent

Working in the $\zeta$-function regularisation scheme, we find certain infinite series representations of the logarithms of massive scalar determinants, $\det(\Delta+m^{2})$ for arbitrary $m^2$, on finite round disks of constant curvature ($R=\frac{2\eta}{L^2}, \eta=0,\pm1$) with Neumann boundary conditions. The derivation of these representations relies on a relation between the Neumann determinants on the disks and the corresponding Dirichlet determinants via the determinants of the Dirichlet-to-Neumann maps on the boundaries of the disks. We corroborate the results in an appendix by computing the Neumann determinants in an alternative way. In the cases of disks with nonzero curvatures, we show that the infinite series representations reduce to  exact expressions for some specific values of $m^2$, viz. $m^2=-\frac{\eta}{L^2}q(q+1)$ with $q\in \mathbb {N}$. Our analysis uses and extends the results obtained in \cite{ChauFerdet} for similar Dirichlet determinants on constant curvature disks.

\vfill

\medskip

\newpage\pagestyle{plain}
\baselineskip 16pt
\setcounter{footnote}{0}

\tableofcontents

\section{\label{intro}Introduction}

The evaluation of functional determinants of differential operators such as Laplacians (and their generalisations) defined on different manifolds has a long and rich history. Such determinants show up in a variety of physical problems including problems in quantum field theory and quantum gravity \cite{Dunne, DHokerPhong1}.  For instance, in case of a free massive scalar field on a closed $d$-dimensional Riemannian manifold, the partition function is given by
\be
\label{partition fn expr.}
Z=\frac{1}{\sqrt{\det(\Delta+m^2)}},
\ee
where $\Delta$ is the positive Laplacian for the background metric and $m$ is the mass of the scalar field.  Computing such a determinant requires working in some regularisation scheme. The $\zeta$-function regularisation \cite{RaySinger, DowkerCritchley, Hawking} is  a convenient scheme for doing such computations and it has been employed extensively in the literature (see, for example \cite{Sarnak, DHokerPhong2, Voros, bfref}). When the manifold has boundaries, one needs to impose some boundary conditions on the field. For Dirichlet or Neumann boundary conditions, the partition function is given by the same expression as in \eqref{partition fn expr.} and the determinant in this expression can still be evaluated in the $\zeta$-function regularisation scheme, but its value depends on the specific boundary conditions.

Such determinants on manifolds with boundaries have been explored for spaces with negative curvature in the infinite volume limit, particularly in the context of holography (see, for example, \cite{CamporesiHiguchi, Helgason, BorthwickJudgePerry, DFPRS} and their references). At finite volume, they are, in general, difficult to evaluate. However, in two dimensions,  when  $m^2=0$, the problem becomes remarkably simple as the determinant is related to a conformal invariant \cite{Weis}. There are also several results for massive determinants on two dimensional Riemannian surfaces with geodesic boundaries (see,  for instance, \cite{GuillarmouGuillope, KMW}).
When the boundary of the surface is not a geodesic and the mass is nonzero, evaluating the determinants is typically challenging. Nonetheless, for sufficiently simple geometries one can make progress.

 A particular instance of such simple geometries is a round disk $\disk$ with constant bulk curvature which is described by the  metric
\be\
\label{const curv metric}
ds_\eta^2=e^{2\sigma_\eta}(dr^2+r^2d\theta^2),\ r\in[0,1],\ \theta\sim\theta+2\pi,
\ee
where  $\eta$ can take the values $0$, $+1$ or $-1$, and accordingly,
\be
\label{conformal factors for metrics}
 \begin{split}
& e^{2\sigma_0}=\Big(\frac{\ell_0}{2\pi}\Big)^2,\\
& e^{2\sigma_+}=\frac{4r_0^2L^2}{(1+ r_0^2 r^2)^2}\ \text{with}\ r_0\in(0,\infty),\\
& e^{2\sigma_-}=\frac{4r_0^2L^2}{(1- r_0^2 r^2)^2}\ \text{with}\ r_0\in(0,1).
\end{split}
\ee
Here $\ell_0$ and $L$ are positive quantities with the dimension of length. The Neumann determinants on these disks will be the objects of our study in this work. We  provide a few geometrical properties of these disks below that will be useful for the analysis done later in the paper. The length of the boundary ($\partial\disk$) of the disk  is given by
\be
 \begin{split}
& \ell_\eta=\begin{cases}\ell_0\ \ \ \ \text{for}\ \eta=0,\\ \frac{4\pi r_0 L}{1+\eta r_0^2 }\ \text{for}\ \eta=\pm 1, \end{cases}
\end{split}
\ee
while the area of the disk is given by
\be
\label{area of the disks}
 \begin{split}
& A_\eta=\begin{cases}\frac{\ell_0^2}{4\pi}\ \ \ \ \  \text{for}\ \eta=0,\\ \frac{4\pi r_0^2 L^2}{1+\eta r_0^2}\ \text{for}\ \eta=\pm 1. \end{cases}
\end{split}
\ee
The constant bulk curvature scalar of the disk is $R_\eta=2\eta/L^2$, while the extrinsic curvature at the boundary $\partial \disk$ is
\be 
\label{extrinsic curv: values}
k_\eta=\begin{cases}\frac{2\pi}{\ell_0}\ \ \ \   \text{for}\ \eta=0,\\\frac{1-\eta r_0^2}{2 r_0 L}\  \text{for}\ \eta=\pm 1.\end{cases}
\ee
We will denote the positive Laplacian corresponding to the disk by $\Delta_\eta$, and it is given by
\be 
\label{extrinsic curv: values}
\Delta_\eta=-e^{-2\sigma_\eta}\Big(\partial_r^2+\frac{1}{r}\partial_r+\frac{1}{r^2}\partial_\theta^2\Big).
\ee

For the above constant curvature disks, the Dirichlet determinant ${\det}_{\text D}(\Delta_\eta+m^2)$ was evaluated in \cite{ChauFerdet} using $\zeta$-function regularisation. For general values of the mass, the logarithm of this determinant was expressed as a convergent infinite series. Furthermore, for some special values of $m^2$, viz, $m^2=-\frac{\eta}{L^2} q(q+1)$ with $q\in\mathbb{N}$, exact expressions were obtained in terms of  certain elementary functions and the Euler Gamma function. These expressions were found to be useful in the analysis of Liouville and Jackiw-Teitelboim gravities \cite{Loopcalc}.

In this paper we will extend the above results by computing the Neumann determinants of $(\Delta_\eta+m^2)$ on the constant curvature round disks defined by the metrics given in \eqref{const curv metric}. To do such a computation, one may follow the strategy discussed in \cite{ChauFerdet} while adapting it to the altered boundary condition, i.e. Neumann instead of Dirichlet. This would involve exploiting the rotational symmetry of the geometry to first find the determinants of Sturm-Liouville operators corresponding to each Fourier mode (for the angular coordinate $\theta$) of the scalar field, and then obtaining the determinant of $(\Delta_\eta+m^2)$ from these Sturm-Liouville determinants. 
We will present this procedure in an appendix. However, in the main text of this paper we will follow a different approach that is somewhat shorter. This approach relies on a relation  between the Neumann and Dirichlet determinants which involves the determinant of the Dirichlet-to-Neumann map on the boundary of the disk. Such relations were recently derived  in \cite{KirstLee1} for more general setups (in arbitrary dimensions and for general metrics). We will show that in the case of the constant curvature disks this relation is particulary simple. Given this relation and the fact that the Dirichlet determinants have already been evaluated in \cite{ChauFerdet}, the computation of the Neumann determinants would  require just the evaluation of the determinants of the Dirichlet-to-Neumann maps. Performing this computation, we will derive an infinite series representation of $\ln{\det}_{\text N}(\Delta_\eta+m^2)$ for general values of $m^2$, and verify that this matches with the result obtained by following the strategy laid out in \cite{ChauFerdet}. In case of the curved disks $(\eta=\pm 1)$, we will then obtain exact expressions of the Neumann determinants for $m^2=-\frac{\eta}{L^2}q(q+1),\ q\in\mathbb{N}$, that are analogous to those obtained for the Dirichlet determinants in \cite{ChauFerdet}.

In an appendix we will also consider an example where the determinant with arbitrary mass can be evaluated exactly. This example corresponds to the positive curvature disk that is a hemisphere. In this case the boundary is a geodesic, and hence we can make a connection with the results in existing literature (in particular, \cite{KMW}). We will derive a very simple expression for the Neumann determinant (for arbitrary $m^2$) on the hemisphere in terms of the Barnes $G$-function that is analogous to the expression of the Dirichlet determinant on the hemisphere evaluated in \cite{ChauFerdet}. We will also show that this determinant vanishes for the special values of $m^2$ that we mentioned above, i.e. for $m^2=-\frac{1}{L^2}q(q+1),\ q\in\mathbb{N}$. 

The organisation of the paper is as follows:

In section \ref{sec: rel between Neumann and Dirichlet dets} we introduce the aforementioned relation between the Neumann and the Dirichlet determinants via the determinant of the Dirichlet-to-Neumann map. In section \ref{sec: det DN map} we evaluate the determinants of the Dirichlet-to-Neumann maps for the constant curvature disks. In section \ref{sec: Determinant for general masses}, we put together everything and use the Dirichlet determinants derived in \cite{ChauFerdet} to obtain the infinite series representations of the logarithms of the Neumann determinants for arbitrary masses.  In section \ref{sec: Neumann det special masses} we derive the aforementioned exact expressions for the Neumann determinants on the disks with nonzero curvature for the special values of  the mass, i.e. $m^2=-\frac{\eta}{L^2}q(q+1),\ q\in\mathbb{N}$. In appendix \ref{appendix: positivity of operators}, for the sake of completeness, we review some arguments for the positivity of the Laplacian with  Neumann and Dirichlet boundary conditions, as well as the positivity of the Dirichlet-to-Neumann map. In appendix \ref{lowest ev of DN map}, we derive the behaviour of the lowest eigenvalue of the Dirichlet-to-Neumann map as $m^2\rightarrow 0$ which is required for the discussion in section \ref{sec: rel between Neumann and Dirichlet dets}. In appendix \ref{app: alternative approach to Neumann det} we verify the infinite series representations obtained in section \ref{sec: Determinant for general masses} through an alternative derivation of the Neumann determinants. In appendix \ref{appendix: hemisphere}, we find an exact expression of the Neumann determinant on a hemisphere for arbitrary mass and show that it vanishes when $m^2=-\frac{1}{L^2}q(q+1)$ with $q\in\mathbb{N}$. 

\section{Relation between the Neumann and Dirichlet determinants via the Dirichlet-to-Neumann map}

\label{sec: rel between Neumann and Dirichlet dets}

In this section we will introduce the result derived in \cite{KirstLee1}  that relates the difference between the Neumann determinant and the Dirichlet determinant to the determinant of the  Dirichlet-to-Neumann map for manifolds with a boundary in arbitrary dimensions.  We will then show that this relation is particularly simple for the manifolds of our interest, viz.  two-dimensional disks of constant curvature. 

\subsection{The relation for general values of the mass}

Consider a $d$-dimensional smooth compact Riemannian manifold $\mathscr{M}$ with a connected boundary $\partial \mathscr{M}$. Let us denote the metric on $\mathscr{M}$ by $g_{\mu\nu}$ and the induced metric on $\partial \mathscr{M}$ by $h_{ab}$. The positive Laplacian corresponding to $g_{\mu\nu}$ is $\Delta=-\frac{1}{\sqrt{g}}\partial_\mu(\sqrt{g}g^{\mu\nu}\partial_\nu)$. The eigenvalues of $\Delta$ would depend on the boundary condition. Accordingly, we denote the eigenvalues for Dirichlet and Neumann boundary conditions by $\la_i^{(\text D)}$ and $\la_i^{(\text N)}$ respectively, with the index $i$ running over all non-negative integers. These eigenvalues are ordered such that $\la_i^{(\text D)}\leq \la_{i+1}^{(\text D)}$ and $\la_i^{(\text N)}\leq \la_{i+1}^{(\text N)}$. Note that there can be degeneracies in the eigenvalues and in that case, the same value will  be assigned to multipe $\la_{i+1}^{(\text D/\text N)}$'s. 

Using these eigenvalues, one can define the following zeta functions:
\be
\label{zeta fn def: Laplacian Dirchlet and Neumann}
\zeta^{(\text D)}(m^2;s)=\sum_{i=0}^\infty\frac{1}{(\la_i^{(\text D)}+m^2)^s},\ \zeta^{(\text N)}(m^2;s)=\sum_{i=0}^\infty\frac{1}{(\la_i^{(\text N)}+m^2)^s}.
\ee
These series representations of the above $\zeta$-functions are valid only when the real part of $s$ is sufficiently large, but the functions themselves can be meromorphically extended to the rest of the complex plane, with these  extended functions being regular at $s=0$. We will implicitly assume this to be true also for all the other $\zeta$-functions that will be introduced later in this paper.

The $\zeta$-regularised Dirichlet and Neumann determinants of $(\Delta+m^2)$ are  defined in terms of the above $\zeta$-functions by the following relations:
\be
\label{Dirichlet Neuman det def}
\ln {\det}_{\text D} (\Delta+m^2)=-\partial_s \zeta^{(\text D)}(m^2;s)|_{s=0},\ \ln {\det}_{\text N} (\Delta+m^2)=-\partial_s \zeta^{(\text N)}(m^2;s)|_{s=0}.
\ee

Now, let us introduce the notion of the Dirichlet-to-Neumann map which is a pseudo-differential operator acting on functions at the boundary. Suppose $f$ is a smooth function defined on $\partial \mathscr{M}$. We can define  a smooth extension of this function to the bulk, which we denote by $H_{m^2}[f]$, that satisfies the following conditions:
\be
(\Delta+m^2)H_{m^2}[f]=0,\ H_{m^2}[f]|_{\partial \mathscr{M}}=f.
\ee
The Dirichlet-to-Neumann map, denoted by $\mathscr{O}_{m^2}$, is defined via the following action on $f$:
\be
\mathscr{O}_{m^2}[f]=\partial_n H_{m^2}[f]|_{\partial \mathscr{M}},
\ee
where $\partial_n$ is the derivative along the outward-pointing unit normal at the boundary. Suppose the eigenvalues of this operator are $v_i^{(m^2)}$, $i\in\mathbb{N}_0$, which we again choose to be arranged in an ascending order. We can define a $\zeta$-function with these eigenvalues as follows:
\be
\label{zeta fn def: DN map}
\zeta^{(\text{DN})}(m^2;s)=\sum_{i=0}^\infty \frac{1}{(v_i^{(m^2)})^s}.
\ee
The $\zeta$-regularised determinant of the Dirichlet-to-Neumann map $\mathscr{O}_{m^2}$ is then defined by
\be
\label{DN map det def}
\ln {\det}(\mathscr{O}_{m^2})=-\partial_s\zeta^{(\text{DN})}(m^2;s)|_{s=0}.
\ee

Now, the relation between the determinants defined in \eqref{Dirichlet Neuman det def} and \eqref{DN map det def}  that was proven in \cite{KirstLee1} is as follows:
\be
\label{DirNeurel:gen dim}
\ln {\det}_{\text N} (\Delta+m^2)- \ln {\det}_{\text D} (\Delta+m^2)=\sum_{j=0}^\nu (m^2)^j b_j+\ln {\det}(\mathscr{O}_{m^2}),
\ee
where $\nu=[\frac{d-1}{2}]$.\footnote{$[x]$ stands for the greatest integer that is less than or equal to $x$.} Here the coefficients $b_j$ are independent of $m^2$, and depend only on the geometry of the space. They are determined by the coefficients in the expansions of the Dirichlet and Neumann heat kernels corresponding to the bulk Laplacian $\Delta$ as well as the coefficients in the large $m$-expansion of the determinant of the Dirichlet-to-Neumann map \cite{KirstLee2}. The physical significance of relations of the form \eqref{DirNeurel:gen dim} will be discussed in \cite{ChauFerNeufree}.

Note that when $d=2$, which is the case of our interest, we have $\nu=0$. Thus, for two-dimensional manifolds with connected boundaries, the relation \eqref{DirNeurel:gen dim} reduces to
\be
\label{DirNeurel:two dim}
\ln {\det}_{\text N} (\Delta+m^2)- \ln {\det}_{\text D} (\Delta+m^2)=b_0+\ln {\det}(\mathscr{O}_{m^2}).
\ee
As we mentioned earlier, $b_0$ is independent of the mass. So, it can be evaluated by comparing the determinants in \eqref{DirNeurel:two dim}  as $m^2\rightarrow0$. .This massless limit was also considered in \cite{Wentworth} and a general form of the quantity $b_0$ was found for arbitrary two-dimensional surfaces with a boundary. In what follows, we will present a derivation of $b_0$ for the constant curvature disks and verify that our result matches with that of \cite{Wentworth}.

\subsection{The massless limit}

To consider the massless limit, let us first note that for the Neumann boundary condition, the lowest eigenvalue of $\Delta_\eta$ is zero\footnote{The other eigenvalues are positive-definite as shown in appendix \ref{appendix: positivity of operators}.}, and the corresponding eigenfunction is the constant function on the disk $\disk$. Therefore, in this case, the lowest  eigenvalue of the operator $(\Delta_\eta+m^2)$ is $m^2$, and  it  goes to zero as  $m^2\rightarrow 0$. Similarly, as  argued in appendix \ref{lowest ev of DN map}, the lowest eigenvalue of $\mathscr{O}_{m^2}$ goes to zero when $m^2\rightarrow 0$ and for small values of $m^2$, it is given by  
\be
v_0^{(m^2)}=\frac{m^2 A_\eta}{\ell_\eta}  +O(m^4).
\ee
The vanishing of these eigenvalues in the $m^2\rightarrow 0$ limit means that $\ln {\det}_{\text N} (\Delta_\eta+m^2)$ and $\ln {\det}(\mathscr{O}_{m^2})$ diverge in this limit. However, one can regularise these divergences by removing the contributions of $\ln(\la_0^{(\text N)}+m^2)=\ln(m^2)$ and $\ln v_0^{(m^2)}\approx\ln\Big(\frac{m^2 A_\eta}{\ell_\eta}\Big)$ as shown below:
\be
\label{reg Neuman DN det def}
\begin{split}
& \ln {\det}_{\text N}'(\Delta_\eta)=\lim_{m^2\rightarrow 0}\Big[\ln {\det}_{\text N} (\Delta_\eta+m^2)-\ln(m^2)\Big],\\
 & \ln {\det}'(\mathscr{O}_{0})=\lim_{m^2\rightarrow 0}\Big[ \ln {\det}(\mathscr{O}_{m^2})-\ln\Big(\frac{m^2 A_\eta}{\ell_\eta}\Big)\Big].
 \end{split}
\ee
This is actually equivalent to removing the lowest eigenvalues  while defining the Neumann $\zeta$-function  for the Laplacian (see \eqref{zeta fn def: Laplacian Dirchlet and Neumann}) and the $\zeta$-function for the Dirichlet-to-Neumann map (see \eqref{zeta fn def: DN map}) in the massless case. The quantities in \eqref{reg Neuman DN det def} are then given by the derivatives of these $\zeta$-functions just as shown in \eqref{Dirichlet Neuman det def} and \eqref{DN map det def}.

Let us note that the difference of $\ln {\det}_{\text N} (\Delta_\eta+m^2)$ and $\ln {\det}(\mathscr{O}_{m^2})$ has a well-defined limit as $m^2\rightarrow 0$ which is given by
\be
\lim_{m^2\rightarrow 0}\Big[\ln {\det}_{\text N} (\Delta_\eta+m^2)- \ln {\det}(\mathscr{O}_{m^2})\Big]= \ln {\det}_{\text N}'(\Delta_\eta)- \ln {\det}'(\mathscr{O}_{0})-\ln\Big(\frac{A_\eta}{\ell_\eta}\Big).
\ee
Moreover, the logarithm of Dirichlet determinant, $\ln {\det}_{\text D} (\Delta_\eta+m^2)$,  has a well-defined limit as $m^2\rightarrow 0$ because the Dirichlet boundary condition does not allow for a nonzero constant mode and consequently there is no zero eigenvalue of the Laplacian. 

Combining all of the above  observations, we can take the limit $m^2\rightarrow 0$ in  equation \eqref{DirNeurel:two dim} to obtain the following expression for $b_0$:
\be
\label{b0 for const curv disks}
b_0=\ln {\det}_{\text N}' (\Delta_\eta)- \ln {\det}_{\text D} (\Delta_\eta)-\ln {\det}'(\mathscr{O}_{0})-\ln \Big(\frac{ A_\eta}{\ell_\eta}\Big).
\ee
 To evaluate this quantity for the different values of $\eta$, let us note that each of these constant curvature disks is related by a Weyl transformation to a flat disk of unit radius as is manifest from the form of the metrics given in \eqref{const curv metric}. We will indicate all quantities related to this flat disk of unit radius with a subscript $\delta$.  Now, in \cite{Weis} it was shown that there are two Weyl-invariant quantities which can be constructed out of the massless determinants ${\det}_{\text N}'(\Delta)$ and ${\det}_{\text D}(\Delta)$. In our case, these two Weyl-invariants imply
 \be
 \label{Weyl-invariant Neumann}
 \begin{split}
 &\ln {\det}_{\text N}' (\Delta_\eta)-\ln(A_\eta)+\frac{S_L[\sigma_\eta]}{12\pi}-\frac{1}{4\pi}\oint_{\partial\disk}ds_\eta \ k_\eta\\
 &= \ln {\det}_{\text N}' (\Delta_\delta)-\ln(A_\delta)-\frac{1}{4\pi}\oint_{\partial\disk}ds_\delta \ k_\delta,\\
 \end{split}
 \ee
 and
  \be
 \label {Weyl-invariant Dirichlet}
 \ln {\det}_{\text D}(\Delta_\eta)+\frac{S_L[\sigma_\eta]}{12\pi}+\frac{1}{4\pi}\oint_{\partial\disk}ds_\eta \ k_\eta= \ln {\det}_{\text D} (\Delta_\delta)+\frac{1}{4\pi}\oint_{\partial\disk}ds_\delta \ k_\delta,
 \ee
 where $S_L[\sigma_\eta]$ is the Liouville action defined as
 \be
 \label{Liouville action}
 S_L[\sigma_\eta]=\int_{\disk} d^2x\ \partial_\mu\sigma_\eta\partial_\mu\sigma_\eta+2\oint_{\partial \disk}ds_\delta\ \sigma_\eta =2\pi\int_0^1 dr\ r(\partial_r\sigma_\eta)^2+4\pi\sigma_\eta|_{r=1}.
 \ee
 Here, $ds_\eta$ and $ds_\delta$ are the  measures for the arc-length at the boundary $\partial\disk$ for the respective metrics so that  $\oint_{\partial\disk}ds_\eta=\ell_\eta$ and $\oint_{\partial\disk}ds_\delta=2\pi$. $k_\eta$ is the extrinsic curvature at the boundary of the constant curvature disk. Its value is given in \eqref{extrinsic curv: values}. Similarly, $k_\delta=1$ is the  extrinsic curvature at the boundary of the flat disk with unit radius.  $A_\delta=\pi$ is the area  of the flat disk with unit radius.
 
 Combining the relations in \eqref{Weyl-invariant Neumann} and \eqref{Weyl-invariant Dirichlet}, we get
 \be
 \label{Weyl-invariant bulk}
 \begin{split}
\ln {\det}_{\text N}' (\Delta_\eta)- \ln {\det}_{\text D} (\Delta_\eta)-\ln A_\eta-\frac{1}{2\pi}\oint_{\partial\disk}ds_\eta \ k_\eta\\
=\ln {\det}_{\text N}' (\Delta_\delta)- \ln {\det}_{\text D} (\Delta_\delta)-\ln A_\delta-\frac{1}{2\pi}\oint_{\partial\disk}ds_\delta \ k_\delta.
\end{split}
\ee
 The values of ${\det}_{\text N}' (\Delta_\delta)$ and ${\det}_{\text D}(\Delta_\delta)$ were computed in \cite{Weis}, and they are given by
\be
\label{massless Neumann det: flat unit disk}
\ln {\det}_{\text N}' (\Delta_\delta)=-\frac{1}{6}\ln(2)+\frac{1}{2}\ln(\pi)-2\zeta_{\text R}'(-1)+\frac{7}{12},
\ee
\be
\label{massless Dirichlet det: flat unit disk}
\ln {\det}_{\text D}(\Delta_\delta)=-\frac{1}{6}\ln(2)-\frac{1}{2}\ln(\pi)-2\zeta_{\text R}'(-1)-\frac{5}{12},
\ee 
where $\zeta_{\text R}'(-1)$ denotes the derivative of the Riemann zeta function evaluated at $-1$. Combining these expressions and subtracting the values of $\ln A_\delta$ and $\frac{1}{2\pi}\oint_{\partial\disk}ds_\delta \ k_\delta$, we get
 \be
\ln {\det}_{\text N}' (\Delta_\delta)- \ln {\det}_{\text D} (\Delta_\delta)-\ln A_\delta-\frac{1}{2\pi}\oint_{\partial\disk}ds_\delta \ k_\delta=0.
\ee
Hence, from \eqref{Weyl-invariant bulk}, we find that
 \be
 \label{diff of massless Neumann and Dirichlet dets}
\ln {\det}_{\text N}' (\Delta_\eta)- \ln {\det}_{\text D} (\Delta_\eta)-\ln A_\eta=\frac{1}{2\pi}\oint_{\partial\disk}ds_\eta \ k_\eta.
\ee

Finally, we also need the determinant of the Dirichlet-to-Neumann map $\mathscr{O}_0$ to complete the computation of $b_0$ given in \eqref{b0 for const curv disks}. For this determinant, we use the general result proven in \cite{GuillarmouGuillope} that in case of 2-dimensional surfaces with a connected boundary, the ratio of the determinant of the Dirichlet-to-Neumann map (${\det}'(\mathscr{O}_0)$) with the length of the boundary is a Weyl-invariant. It was also shown in \cite{EdwardWu, GuillarmouGuillope} that for the conformal class of the flat disk, this Weyl-invariant is just $1$. Hence,  for the constant curvature disks of our interest, we have
 \be
 \label{zero mass det DN map}
\ln{\det}'(\mathscr{O}_0)=\ln\ell_\eta.
\ee
We will verify this explicitly in section \ref{sec: det DN map}.

Using the equations \eqref{diff of massless Neumann and Dirichlet dets} and \eqref{zero mass det DN map}, we get the following value of $b_0$ from \eqref{b0 for const curv disks}:
\be
 \label{b0 for const curv disks final expr}
b_0=\frac{1}{2\pi}\oint_{\partial\disk}ds_\eta \ k_\eta=\frac{k_\eta\ell_\eta}{2\pi}=\begin{cases}1\ \ \ \ \ \ \ \ \ \  \text{for}\ \eta=0,\\1-\frac{\eta A_\eta}{2\pi L^2 }\  \text{for}\ \eta=\pm 1.\end{cases}
\ee
We note that this is exactly the value of $b_0$ that can be obtained from Proposition 2.2 in \cite{Wentworth}.

\section{Evaluation of the determinant of the Dirichlet-to-Neumann map}

\label{sec: det DN map}

Having introduced the relation \eqref{DirNeurel:two dim} and computed the quantity $b_0$ for the constant curvature disks, let us now evaluate the determinant of the Dirichlet-to-Neumann map $\mathscr{O}_{m^2}$.

\subsection{Flat case}

Let us first consider the flat case, i.e. $\eta=0$. To find the eigenvalues of the operator $\mathscr{O}_{m^2}$ in this case, we need to find (real) solutions to the Steklov problem,
\be
\label{Steklov problem: flat case}
(\Delta_0+m^2)\Phi=0,\ \partial_n\Phi|_{\partial\disk}=\varv\Phi|_{\partial\disk},
\ee
for some $\varv\geq 0$. The  normal derivative at the boundary and the Laplacian are given in the polar coordinates introduced in \eqref{const curv metric} by
\be
\partial_n=e^{-\sigma_0}\partial_r\Big|_{\partial\disk}=\frac{2\pi}{\ell_0}\partial_r\Big|_{\partial\disk},\ \Delta_0=-e^{-2\sigma_0}\Big[\partial_r^2+\frac{1}{r}\partial_r+\frac{1}{r^2}\partial_\theta^2\Big]=-\frac{4\pi^2}{\ell_0^2}\Big[\partial_r^2+\frac{1}{r}\partial_r+\frac{1}{r^2}\partial_\theta^2\Big].
\ee

To find the solutions to the above problem, let us decompose the field $\Phi$ into its Fourier modes:
\be
\label{Fourier decomp.: bulk field}
\Phi(r,\theta)=\sum_{j=-\infty}^\infty\Phi_j(r)e^{ij\theta}.
\ee
Since we are looking for real solutions, we demand that $\Phi_j(r)^*=\Phi_{-j}(r)$. Moreover, let  $\varphi(\theta)=\Phi(1,\theta)$ be the boundary value of $\Phi$. The Fourier decomposition of this boundary field is $\varphi$
\be
\label{Fourier decomp.: bdy field}
\varphi(\theta)=\sum_{j=-\infty}^\infty\varphi_j e^{ij\theta},
\ee
where $\varphi_j=\Phi_j(1)$. Then the Steklov problem posed in \eqref{Steklov problem: flat case} reduces to seeking solutions to the equation
\be
\label{Steklov problem Fourier mode diff eqn: flat case}
\Phi_j''(r)+\frac{1}{r}\Phi_j'(r)-\frac{j^2}{r^2}\Phi_j(r)=\frac{m^2\ell_0^2}{4\pi^2}\Phi_j(r),
\ee
with the boundary condition
\be
\label{Steklov problem Fourier mode bdy cond.: flat case}
 \frac{2\pi}{\ell_0}\Phi_j'(1)=\varv\varphi_j,\ \text{for all}\ j\in\mathbb{Z}.
\ee
The solution to the  differential equation \eqref{Steklov problem Fourier mode diff eqn: flat case} that is regular at $r=0$ is
\be
\Phi_j(r)=c_j  J_j\Big(\frac{i m\ell_0}{2\pi}r\Big),
\ee
for some $c_j\in\mathbb {C}$ with the reality of $\Phi$ requiring $c_j^*=c_{-j}$ when $m^2> 0$. Here $J_j$ is the Bessel function of order $j$. The boundary condition \eqref{Steklov problem Fourier mode bdy cond.: flat case}  can  be satisfied by this solution   only if
\be
\label{constraints on coefficients}
c_j\propto(\delta_{j,p}+\delta_{j,-p})\ \text{or}\ c_j\propto i(\delta_{j,p}-\delta_{j,-p})
\ee
for some $p\in\mathbb{Z}$ with the proportionality constant being a real number. Accordingly, the independent real solutions to the Steklov problem can be enumerated as follows:
\be
\Phi^{(p)}(r,\theta)
=\begin{cases}i^p J_p\Big(\frac{i m\ell_0}{2\pi}r\Big)\cos(p\theta) \ \text{when}\ p\geq 0,\\ i^p J_p\Big(\frac{i m\ell_0}{2\pi}r\Big)\sin(p\theta) \ \text{when}\ p< 0,\end{cases}\\
\ee
with $p\in\mathbb{Z}$. The corresponding values of $\varv$ are given by
\be
\varv^{(p)}(m^2)= \frac{2\pi}{\ell_0}\frac{\frac{d}{dr}[ J_p(\frac{i m\ell_0}{2\pi}r)]|_{r=1}}{J_p(\frac{i m\ell_0}{2\pi})}=\frac{2\pi}{\ell_0}\Big[-p+\frac{i m\ell_0}{2\pi}\frac{ J_{p-1}(\frac{i m\ell_0}{2\pi})}{J_p(\frac{i m\ell_0}{2\pi})}\Big].
\ee
One can check that $\varv^{(p)}(m^2)$ is symmetric under $p\rightarrow-p$ by using the identities $J_{-n}(z)=(-1)^n J_n(z)$ and $J_n(z)=\frac{2(n-1)}{z}J_{n-1}(z)-J_{n-2}(z)$ that are valid for any $n\in\mathbb{Z}$. Therefore, we can express $\varv^{(p)}(m^2)$ as
\be
\varv^{(p)}(m^2)
=\frac{2\pi}{\ell_0}\Big[-|p|+\frac{i m\ell_0}{2\pi}\frac{ J_{|p|-1}(\frac{i m\ell_0}{2\pi})}{J_{|p|}(\frac{i m\ell_0}{2\pi})}\Big].
\ee
These are the eigenvalues of the Dirichlet-to-Neumann map $\mathscr{O}_{m^2}$. They can be re-arranged in an ascending order to express them as $v_k^{(m^2)}$ with $k\geq 0$, the notation we had introduced earlier. However, from now on we will stick to the notation $\varv^{(p)}(m^2)$ with  $p\in \mathbb{Z}$ as it is more convenient for the present discussion.

Let us note here that massless limits of these eigenvalues are readily obtained as follows:
\be
\varv^{(p)}(0)
=\frac{2\pi}{\ell_0}\Big[-|p|+\lim_{x\rightarrow 0}\Big\{ix\frac{ J_{|p|-1}(i x)}{J_{|p|}(ix)}\Big\}\Big]=\frac{2\pi}{\ell_0}\Big[-|p|+2|p|\Big]=\frac{2\pi}{\ell_0}|p|.
\ee
From the eigenvalues of $\mathscr{O}_{m^2}$ and $\mathscr{O}_{0}$ given above, one can obtain the following infinite series representation of the logarithm of the ratio of the determinants of these operators:
\be
\label{series rep. DN map: flat}
\ln\Big(\frac{{\det}(\mathscr{O}_{m^2})}{{\det}'(\mathscr{O}_0)}\Big)=\ln\Big(\varv^{(0)}(m^2)\Big)+2\sum_{p=1}^\infty\ln\Big(\frac{\varv^{(p)}(m^2)}{\varv^{(p)}(0)}\Big).
\ee
The convergence of this series can be checked by using the asymptotic  expansion of the Bessel function at large order which is 
\be
\begin{split}
J_\nu(z)=\frac{1}{\sqrt{2\pi}}\exp\Big[\nu+\nu\ln(\frac{z}{2})-(\nu+\frac{1}{2})\ln\nu\Big]\Big[1-\frac{3z^2+1}{12\nu}+O(\frac{1}{\nu^2})\Big].
\end{split}
\ee
The above asymptotic expansion of the Bessel function yields $\frac{\varv^{(p)}(m^2)}{\varv^{(p)}(0)}=1+O(1/p^2)$, and hence, $\ln\Big(\frac{\varv^{(p)}(m^2)}{\varv^{(p)}(0)}\Big)=O(1/p^2)$ as $p\rightarrow\infty$. This ensures the convergence of the series in \eqref{series rep. DN map: flat}.
Substituting the values of $\varv^{(p)}(m^2)$ and $\varv^{(p)}(0)$ in \eqref{series rep. DN map: flat}, and then using the relation $I_\alpha(x)=i^{-\alpha}J_\alpha(i x)$ where $I_\alpha$ is a modified Bessel function of the first kind, we get
\be
\label{ratio of massive and massless det DN map: flat case}
\ln\Big(\frac{{\det}(\mathscr{O}_{m^2})}{{\det}'(\mathscr{O}_0)}\Big)
=\ln\Big(m\frac{ I_{-1}(\frac{m\ell_0}{2\pi})}{I_{0}(\frac{m\ell_0}{2\pi})}\Big)+2\sum_{p=1}^\infty\ln\Big(\frac{ \frac{ m\ell_0}{2\pi} I_{p-1}(\frac{m\ell_0}{2\pi})-p I_{p}(\frac{ m\ell_0}{2\pi})}{p I_{p}(\frac{ m\ell_0}{2\pi})}\Big).
\ee
Furthermore, from \eqref{zero mass det DN map} we know that $\ln{\det}'(\mathscr{O}_0)=\ln\ell_0$. We can now check this explicitly by considering the following $\zeta$-function:
\be
{\zeta^{(\text{DN})}}'(0;s)=\sum_{p\in\mathbb{Z},p\neq 0} (\varv^{(p)}(0))^{-s}=2 (\frac{2\pi}{\ell_0})^{-s}\sum_{p=1}^\infty p^{-s}=2 (\frac{2\pi}{\ell_0})^{-s}\zeta_{\text R}(s),
\ee
where $\zeta_{\text R}$ denotes the Riemann zeta function. This yields the expected result for the determinant ${\det}'(\mathscr{O}_0)$ as shown below:
\be
\label{massless det DN map check:flat case}
\begin{split}
\ln {\det}'(\mathscr{O}_0)
&=-\partial_s {\zeta^{(\text{DN})}}'(0;s)|_{s=0}=-2 \partial_s\Big[(\frac{2\pi}{\ell_0})^{-s}\zeta_{\text R}(s)\Big]\Big|_{s=0}\\
&=-2\Big[\ln(\frac{\ell_0}{2\pi})\zeta_{\text{R}}(0)+\zeta_{\text R}'(0)\Big]=-2\Big[-\frac{1}{2}\ln(\frac{\ell_0}{2\pi})-\frac{1}{2}\ln(2\pi)\Big]=\ln(\ell_0).
\end{split}
\ee
Here we have used the fact that $\zeta_{\text R}(0)=-\frac{1}{2}$ and $\zeta_{\text R}'(0)=-\frac{1}{2}\ln(2\pi)$.

Combining the results obtained in \eqref{ratio of massive and massless det DN map: flat case} and \eqref{massless det DN map check:flat case}, we get the following expression of the determinant of the Dirchlet-to-Neumann map:
\be
\ln {\det}(\mathscr{O}_{m^2})
=\ln(2\pi)+\ln\Big(\frac{\frac{m\ell_0}{2\pi} I_{-1}(\frac{m\ell_0}{2\pi})}{I_{0}(\frac{m\ell_0}{2\pi})}\Big)+2\sum_{p=1}^\infty\ln\Big(\frac{ \frac{ m\ell_0}{2\pi} I_{p-1}(\frac{m\ell_0}{2\pi})-p I_{p}(\frac{ m\ell_0}{2\pi})}{p I_{p}(\frac{ m\ell_0}{2\pi})}\Big).
\ee

\subsection{Curved cases}

Now, let us turn our attention to the curved cases, i.e. $\eta=\pm 1$. Evaluating the determinant of the Dirichlet-to-Neumann map in these cases is analogous to the same for the flat case. The Steklov problem involves finding real solutions to
\be
\label{Steklov problem: curved case}
(\Delta_\eta+m^2)\Phi=0,\ \partial_n\Phi|_{\partial\disk}=\varv\Phi|_{\partial\disk},
\ee
for some $\varv\geq 0$. The  normal derivative at the boundary and the Laplacian are
\be
\begin{split}
&\partial_n=e^{-\sigma_\eta}\partial_r\Big|_{\partial\disk}=\frac{1+ \eta r_0^2 }{2r_0 L}\partial_r\Big|_{\partial\disk},\\
& \Delta_\eta=-e^{-2\sigma_\eta}\Big[\partial_r^2+\frac{1}{r}\partial_r+\frac{1}{r^2}\partial_\theta^2\Big]=-\frac{(1+ \eta r_0^2 r^2)^2}{4r_0^2 L^2}\Big[\partial_r^2+\frac{1}{r}\partial_r+\frac{1}{r^2}\partial_\theta^2\Big].
\end{split}
\ee
To find the solutions to the Steklov problem, one can again decompose the bulk field and its boundary value as shown in \eqref{Fourier decomp.: bulk field} and \eqref{Fourier decomp.: bdy field}. Then the Steklov problem reduces to finding solutions to the ordinary differential equations
\be
\label{Steklov problem Fourier mode diff eqn: curved case}
\Phi_j''(r)+\frac{1}{r}\Phi_j'(r)-\frac{j^2}{r^2}\Phi_j(r)=\frac{4r_0^2 m^2 L^2}{(1+ \eta r_0^2 r^2)^2}\Phi_j(r),
\ee
with the boundary condition
\be
\label{Steklov problem Fourier mode bdy cond.: curved case}
 \frac{1+ \eta r_0^2 }{2r_0 L}\Phi_j'(1)=\varv\varphi_j,\ \ \text{for all}\ \  j\in\mathbb{Z}. 
\ee
The solution to  \eqref{Steklov problem Fourier mode diff eqn: curved case} that is regular at $r=0$ is
\be
\Phi_j(r)=c_j f_j^{(\eta)}(r;-m^2 L^2)
\ee
where
\be
\label{f_n def.}
f_j^{(\eta)}(r;\la)=r^{|j|}F\Big(\frac{1}{2}(1+\sqrt{1+4\eta\la}),\frac{1}{2}(1-\sqrt{1+4\eta\la}),|j|+1,\frac{\eta r_0^2 r^2}{1+\eta r_0^2 r^2}\Big),
\ee
with $F$ representing the Gaussian hypergeometric function $\ _2F_1$. Just as in the flat case, the boundary condition \eqref{Steklov problem Fourier mode bdy cond.: curved case} can be satisfied for all $j\in \mathbb{Z}$ only if the coefficients $\{c_j\}$ are of the forms given in \eqref{constraints on coefficients}. Accordingly, the independent real solutions to the above Steklov problem are 
\be
\Phi^{(p)}(r,\theta)
=\begin{cases}f_p^{(\eta)}(r;-m^2 L^2)\cos(p\theta) \ \text{when}\ p\geq 0,\\ f_p^{(\eta)}(r;-m^2 L^2)\sin(p\theta) \ \text{when}\ p< 0,\end{cases}\\
\ee
with $p\in\mathbb{Z}$. The corresponding values of $\varv$ are given by
\be
\label{eigenvalues of massive DN map: curved disks}
\varv^{(p)}(m^2)=  \frac{1+ \eta r_0^2 }{2r_0 L}\frac{g_p^{(\eta)}(-m^2 L^2)}{f_p^{(\eta)}(1;-m^2 L^2)},
\ee
where
\be
\label{g_n def.}
\begin{split}
&g_p^{(\eta)}(\la)
=\frac{\partial}{\partial r}[ f_p^{(\eta)}(r;\la)]|_{r=1}\\
&=|p|F\Big(\frac{1}{2}(1+\sqrt{1+4\eta\la}),\frac{1}{2}(1-\sqrt{1+4\eta\la}),|p|+1,\frac{\eta r_0^2 }{1+\eta r_0^2}\Big)\\
&\quad-\frac{2 \la r_0^2}{(|p|+1)(1+\eta r_0^2)^2}F\Big(\frac{1}{2}(3+\sqrt{1+4\eta\la}),\frac{1}{2}(3-\sqrt{1+4\eta\la}),|p|+2,\frac{\eta r_0^2}{1+\eta r_0^2 }\Big).
\end{split}
\ee
The quantities $\varv^{(p)}(m^2)$ with $p\in\mathbb{Z}$ are the eigenvalues of $\mathscr{O}_{m^2}$. The massless limits of these eigenvalues are given by
\be
\label{eigenvalues of massless DN map: curved disks}
\varv^{(p)}(0)
=\frac{1+ \eta r_0^2 }{2r_0 L}\frac{g_p^{(\eta)}(0)}{f_p^{(\eta)}(1;0)}=\frac{1+ \eta r_0^2 }{2r_0 L}|p|.
\ee
From these values of $\varv^{(p)}(m^2)$  and $\varv^{(p)}(0)$, we can again obtain an infinite series representation of the logarithm of the ratio of the determinants of $\mathscr{O}_{m^2}$ and $\mathscr{O}_0$ which is of the following form:
\be
\label{series rep. DN map: curved}
\ln\Big(\frac{{\det}(\mathscr{O}_{m^2})}{{\det}'(\mathscr{O}_0)}\Big)=\ln\Big(\varv^{(0)}(m^2)\Big)+2\sum_{p=1}^\infty\ln\Big(\frac{\varv^{(p)}(m^2)}{\varv^{(p)}(0)}\Big).
\ee
The convergence of this series can be checked by using the following asymptotic form of $F(a,b,c+\la,z)$ in the large $\la$ regime \cite{Temme}:
\be
\label{hypgeom asymp large c}
F(a,b,c+\la,z)=\frac{\Gamma(c+\la)}{\Gamma(c+\la-b)}\la^{-b}\Big[1+O(1/\la)\Big].
\ee
This gives us the following asymptotic form of $\frac{\varv^{(p)}(m^2)}{\varv^{(p)}(0)}$ in the large $p$ regime:
\be
\begin{split}
\frac{\varv^{(p)}(m^2)}{\varv^{(p)}(0)}
&=1+\frac{2 m^2 L^2 r_0^2}{(1+\eta r_0^2)^2}\frac{1}{p^2}+O(1/p^3).
\end{split}
\ee
Hence, in this regime $\ln\Big(\frac{\varv^{(p)}(m^2)}{\varv^{(p)}(0)}\Big)=O(1/p^2)$ which ensures the convergence of the series  in \eqref{series rep. DN map: curved}.

Substituting  $\varv^{(p)}(m^2)$ and $\varv^{(p)}(0)$ in \eqref{series rep. DN map: curved} by their values given in \eqref{eigenvalues of massive DN map: curved disks} and \eqref{eigenvalues of massless DN map: curved disks} respectively, we get
 \be
 \label{ratio of massive and massless det DN map: curved case}
\ln\Big(\frac{{\det}(\mathscr{O}_{m^2})}{{\det}'(\mathscr{O}_0)}\Big)=\ln\Big(\frac{1+ \eta r_0^2 }{2r_0 L}\Big)+\ln\Big(\frac{g_0^{(\eta)}(-m^2 L^2)}{f_0^{(\eta)}(1;-m^2 L^2)}\Big)+2\sum_{p=1}^\infty\ln\Big(\frac{g_p^{(\eta)}(-m^2 L^2)}{p f_p^{(\eta)}(1;-m^2 L^2)}\Big).
\ee
Now, from \eqref{zero mass det DN map}, we have that 
 \be
 \label{massless det DN map:curved case}
\ln{\det}'(\mathscr{O}_0)=\ln\ell_\eta=\ln \Big(\frac{4\pi r_0 L}{1+\eta r_0^2 }\Big).
\ee
As in the flat case, we can check this explicitly by defining the $\zeta$-function,
\be
{\zeta^{(\text{DN})}}'(0;s)=\sum_{p\in\mathbb{Z},p\neq 0} (\varv^{(p)}(0))^{-s}=2 \Big(\frac{1+ \eta r_0^2 }{2r_0 L}\Big)^{-s}\zeta_{\text R}(s).
\ee
This yields the expression of $\ln{\det}'(\mathscr{O}_0)$ in \eqref{massless det DN map:curved case} as shown below:
\be
\begin{split}
\ln {\det}'(\mathscr{O}_0)
&=-\partial_s {\zeta^{(\text{DN})}}'(0;s)|_{s=0}=-2 \partial_s\Big[\Big(\frac{1+ \eta r_0^2 }{2r_0 L}\Big)^{-s}\zeta_{\text R}(s)\Big]\Big|_{s=0}\\
&=-2\Big[\ln\Big(\frac{2r_0 L}{1+ \eta r_0^2 }\Big)\zeta_{\text{R}}(0)+\zeta_{\text R}'(0)\Big]=-2\Big[-\frac{1}{2}\ln\Big(\frac{2r_0 L}{1+ \eta r_0^2 }\Big)-\frac{1}{2}\ln(2\pi)\Big]\\
&=\ln \Big(\frac{4\pi r_0 L}{1+\eta r_0^2 }\Big).
\end{split}
\ee

Combining the expressions in \eqref{ratio of massive and massless det DN map: curved case} and \eqref{massless det DN map:curved case}, we get
 \be
 \label{det DN map final result}
\ln {\det}(\mathscr{O}_{m^2})=\ln(2\pi)+\ln\Big(\frac{g_0^{(\eta)}(-m^2 L^2)}{f_0^{(\eta)}(1;-m^2 L^2)}\Big)+2\sum_{p=1}^\infty\ln\Big(\frac{g_p^{(\eta)}(-m^2 L^2)}{p f_p^{(\eta)}(1;-m^2 L^2)}\Big).
\ee

\section{Neumann determinants for general values of the mass}
\label{sec: Determinant for general masses}
Now that we have computed the determinant of the Dirichlet-to-Neumann map $\mathscr{O}_{m^2}$, the only things that we need to finish the calculation of the Neumann determinants ${\det}_{\text N}(\Delta_\eta+m^2)$ are the corresponding Dirichlet determinants ${\det}_{\text D}(\Delta_\eta+m^2)$ . These  Dirichlet determinants were computed in \cite{ChauFerdet}. Here we just quote the results.

In the flat case ($\eta=0$), the logarithm of the  Dirichlet determinant of ($\Delta_0+m^2$) is given by the following infinite series:
\be
\begin{split}
\ln {\det}_{\text D} (\Delta_0+m^2)
&=\Big(-\frac{1}{3}+\frac{A_0m^2}{2\pi}\Big)\ln(\frac{\ell_0}{2\pi})+\frac{1}{3}\ln 2-\frac{1}{2}\ln(2\pi)-\frac{5}{12}-2\zeta_{\text R}'(-1)\\
&\quad+\frac{1}{2}(\gamma-1-\ln 2)(\frac{m\ell_0}{2\pi})^2+\ln \Big(I_0(\frac{m\ell_0}{2\pi})e^{-\frac{1}{4}(\frac{m\ell_0}{2\pi})^2}\Big)\\
&\quad+2\sum_{n=1}^\infty\ln\Big[\frac{2^n n!}{(\frac{m\ell_0}{2\pi})^n}I_n(\frac{m\ell_0}{2\pi})e^{-\frac{(\frac{m\ell_0}{2\pi})^2}{4(n+1)}}\Big],
\end{split}
\ee
where $\gamma$ is the Euler-Mascheroni constant. Similarly, in the curved cases $(\eta=\pm 1)$, the logarithm of the  Dirichlet determinant of ($\Delta_\eta+m^2$) is given by
\be
\begin{split}
\ln {\det}_{\text D} (\Delta_\eta+m^2)
&=\Big(-\frac{1}{3}+\frac{A_\eta m^2}{2\pi}\Big)\ln(L)-\frac{1}{2}\ln(2\pi)-\frac{5}{12}-2\zeta_{\text R}'(-1)\\
&\quad+\frac{\eta A_\eta}{3\pi L^2}-\frac{1}{3}\ln r_0+\frac{A_\eta m^2}{2\pi}(\gamma-1+\ln r_0)\\
&\quad+\ln \Big[f_0^{(\eta)}(1;-m^2L^2)e^{-\frac{A_\eta m^2}{4\pi}F(1,1,2,\frac{\eta A_\eta}{4\pi L^2})}\Big]\\
&\quad+2\sum_{n=1}^\infty\ln \Big[f_n^{(\eta)}(1;-m^2L^2)e^{-\frac{A_\eta m^2}{4\pi(n+1)}F(1,1,n+2,\frac{\eta A_\eta}{4\pi L^2})}\Big],
\end{split}
\ee
where $f_n^{(\eta)}$ is the hypergeometric function defined in \eqref{f_n def.}.

Now, from the relation \eqref{DirNeurel:two dim} and the values of $b_0$ and ${\det}(\mathscr{O}_{m^2})$ obtained in section \ref{sec: rel between Neumann and Dirichlet dets} and section \ref{sec: det DN map} respectively, we can immediately get the desired Neumann determinants. In the flat case ($\eta=0$), the logarithm of the  Neumann determinant of ($\Delta_0+m^2$) is given by
\be
\begin{split}
\ln {\det}_{\text N} (\Delta_0+m^2)
&=\Big(-\frac{1}{3}+\frac{A_0m^2}{2\pi}\Big)\ln(\frac{\ell_0}{2\pi})+\frac{1}{3}\ln 2+\frac{1}{2}\ln(2\pi)+\frac{7}{12}-2\zeta_{\text R}'(-1)\\
&\quad+\frac{1}{2}(\gamma-1-\ln 2)(\frac{m\ell_0}{2\pi})^2+\ln \Big(\frac{m\ell_0}{2\pi} I_{-1}(\frac{m\ell_0}{2\pi}) e^{-\frac{1}{4}(\frac{m\ell_0}{2\pi})^2}\Big)\\
&\quad+2\sum_{n=1}^\infty\ln\Big[\frac{2^n (n-1)!}{(\frac{m\ell_0}{2\pi})^n}\Big( \frac{ m\ell_0}{2\pi} I_{n-1}(\frac{m\ell_0}{2\pi})-n I_{n}(\frac{ m\ell_0}{2\pi})\Big)e^{-\frac{(\frac{m\ell_0}{2\pi})^2}{4(n+1)}}\Big].
\end{split}
\ee
In the curved cases $(\eta=\pm 1)$, the logarithm of the  Neumann determinant of ($\Delta_\eta+m^2$) is given by
\be
\label{Neumann determinants:curved disk}
\begin{split}
\ln {\det}_{\text N} (\Delta_\eta+m^2)
&=\Big(-\frac{1}{3}+\frac{A_\eta m^2}{2\pi}\Big)\ln(L)+\frac{1}{2}\ln(2\pi)+\frac{7}{12}-2\zeta_{\text R}'(-1)\\
&\quad-\frac{\eta A_\eta}{6\pi L^2}-\frac{1}{3}\ln r_0+\frac{A_\eta m^2}{2\pi}(\gamma-1+\ln r_0)\\
&\quad+\ln \Big[g_0^{(\eta)}(-m^2L^2)e^{-\frac{A_\eta m^2}{4\pi}F(1,1,2,\frac{\eta A_\eta}{4\pi L^2})}\Big]\\
&\quad+2\sum_{n=1}^\infty\ln \Big[\frac{g_n^{(\eta)}(-m^2L^2)}{n}e^{-\frac{A_\eta m^2}{4\pi(n+1)}F(1,1,n+2,\frac{\eta A_\eta}{4\pi L^2})}\Big],
\end{split}
\ee
where $g_n^{(\eta)}$ is the function defined in \eqref{g_n def.}. Let us note here that the expression of the determinant in \eqref{Neumann determinants:curved disk} matches exactly with the result obtained in appendix \ref{app: alternative approach to Neumann det} via an alternative approach  (see equation  \eqref{Neumann det gen masses via alt approach}). Let us also note that in appendix \ref{appendix: hemisphere} we provide a particular example where the above series can be evaluated exactly. This example corresponds to a positive curvature disk which is a hemisphere. In appendix \ref{appendix: hemisphere} we show that the determinant in this case has a very simple expression in terms of the Barnes $G$-function.

\section{The special case of $m^2=-\frac{\eta}{L^2} q(q+1)$ with $q\in \mathbb{N}$}

\label{sec: Neumann det special masses}

In \eqref{Neumann determinants:curved disk} we obtained infinite series representations of the logarithms of the Neumann determinants on the curved disks corresponding to $\eta=\pm 1$. We will now show that these quantities reduce to some exact expressions for the special cases when $m^2=-\frac{\eta}{L^2} q(q+1)$ with $q\in \mathbb{N}$. Note that for the positive curvature  disk, i.e. for $\eta=1$, these values of $m^2$ are negative. Although our infinite series representations of the logarithms of the Neumann determinants were derived for positive $m^2$, the expressions of the determinants obtained by taking the exponential of these series can be analytically continued to the domain where $m^2\leq0$. These analytically continued expressions would vanish at $m^2=0$ because of the contribution of the zero eigenvalue of the Laplacian, and they would be negative for sufficiently small negative values of $m^2$.  It is these analytically continued expressions that we would be considering in the positive curvature case.

Having given the above clarification, let us now go on to derive the exact expressions of the Neumann determinants for the  special values of the masses mentioned above. To derive these expressions, it is convenient to define
\be
\label{z_eta def.}
z_\eta=\frac{\eta r_0^2}{1+\eta r_0^2}=\frac{\eta A_\eta}{4\pi L^2},
\ee
and 
\begin{equation}
\begin{split}
&\mathcal{W}_n(q;z_\eta)
=\ln\Bigg[\frac{g_n^{(\eta)}(\eta q(q+1))}{n}e^{\frac{q(q+1)}{n+1}z_\eta  F(1,1,n+2,z_\eta)}\Bigg]\\
&= \frac{q(q+1) }{n+1}z_\eta F(1,1,n+2,z_\eta)\\
&\quad+\ln\Big[  F(1+q,-q,n+1,z_\eta)-2z_\eta(1-z_\eta)\frac{q(q+1)}{n(n+1)}\ F(2+q,1-q,n+2, z_\eta)\Big],
\end{split}
\end{equation}
for $n\geq 1$.
Note that $\sum_{n=1}^\infty \mathcal{W}_n(q;z_\eta)$ is the infinite series that enters in the expression of the Neumann determinant given in \eqref{Neumann determinants:curved disk} when $m^2=-\frac{\eta}{L^2} q(q+1)$. We will first evaluate this series and then substitute it by its value in the expression of the Neumann determinant. Our arguments will  closely follow those given in \cite{ChauFerdet} for evaluating a similar series appearing in the Dirichlet determinant.

\subsection{Evaluation of the series $\sum_{n=1}^\infty \mathcal{W}_n(q;z_\eta)$}

The infinite sum over the quantities $\mathcal{W}_n(q;z_\eta)$  can be expressed as a limit of partial sums as follows:
\begin{equation}
\begin{split}
\sum_{n=1}^\infty\mathcal{W}_n(q;z_\eta)=\lim_{N\rightarrow\infty}\Big[q(q+1)\sum_{n=1}^N  \frac{ z_\eta}{n+1}F(1,1,n+2,z_\eta)+\Sigma_N(q,z_\eta)\Big]
\end{split}
\label{sum over Wn}
\end{equation}
where
\begin{equation}
\begin{split}
&\Sigma_N(q;z_\eta)=\sum_{n=1}^N \ln\Bigg[F(1+q,-q,n+1,z_\eta)\\
&\qquad\qquad\qquad\qquad-2z_\eta(1-z_\eta)\frac{q(q+1)}{n(n+1)}\ F(2+q,1-q,n+2, z_\eta)\Bigg].
\end{split}
\label{defn of SigmaNs}
\end{equation}
The infinite sum $\sum_{n=1}^\infty \mathcal{W}_n(q;z_\eta)$ can be evaluated by first considering the asymptotic form of these partial sums in the large $N$ regime, and then taking the limit $N\rightarrow\infty$. The asymptotic form of the first term of the partial sum in \eqref{sum over Wn} at large $N$ was obtained in \cite{ChauFerdet} (see Lemma 9.1 in \cite{ChauFerdet}), and it is given by
\begin{equation}
\begin{split}
&\sum_{n=1}^N  \frac{ z_\eta}{n+1}F(1,1,n+2,z_\eta)=z_\eta \ln N+\gamma z_\eta+(1-z_\eta)\ln(1-z_\eta)+O(1/N).
\end{split}
\label{asymp. form of first term in sum over Wn}
\end{equation}
For the asymptotic form of the second term of the partial sum in \eqref{sum over Wn}, we prove the following lemma.\footnote{This proof is similar to the proof of Lemma 9.2 in \cite{ChauFerdet}. However, the polynomial whose zeros $\Big\{ \omega_{k}^{(q+1)}(z_\eta): k\in\{1,\cdots,q+1\}\Big\}$ appear in \eqref{asymp form of SigmaN} is distinct from the polynomial whose zeros appeared in Lemma 9.2 of  \cite{ChauFerdet}.}
\begin{lemma} In the large N regime,
\begin{equation}
\label{asymp form of SigmaN}
\begin{split}
&\Sigma_N(q,z_\eta)=-q(q+1)z_\eta\ln N-\sum_{k=1}^{q+1}\ln\Bigg[\frac{\Gamma\Big(1+q- \omega_{k}^{(q+1)}(z_\eta)\Big)}{\Gamma(k)}\Bigg]+O(1/N),
\end{split}
\end{equation}
where $\Big\{ \omega_{k}^{(q+1)}(z_\eta): k\in\{1,\cdots,q+1\}\Big\}$ is the set of zeros of a $(q+1)$-degree polynomial $\mathcal{Q}_{q+1}(x;z_\eta)$ in the variable $x$ that is  defined in \eqref{definition of Q}.
\end{lemma}
\begin{proof}
To prove the above lemma, let us first define the following function:
\begin{equation}
\begin{split}
\mathcal{Q}_{q+1}(x;z_\eta)
=&\frac{\Gamma(x+1)}{\Gamma(x-q+2)}\Big[(x-q)(x-q+1)F(1+q,-q,x-q+1,z_\eta)\\
&\qquad\qquad\qquad-2z_\eta(1-z_\eta)q(q+1)\ F(2+q,1-q,x-q+2, z_\eta)\Big].
\end{split}
\label{definition of Q}
\end{equation}
By using the series representation of hypergeometric functions, one can verify that $F(1+q,-q,x-q+1,z_\eta)$ and $F(2+q,1-q,x-q+2, z_\eta)$ are polynomials of degree $q$ and $q-1$ respectively in the variable $z_\eta$. Consequently, $\mathcal{Q}_{q+1}(x;z_\eta)$ is a polynomial of degree $q+1$ in $z_\eta$ which has the following form:
\begin{equation}
\begin{split}
\mathcal{Q}_{q+1}(x;z_\eta)
=&\sum_{k=0}^{q} \frac{\Gamma(1+2q-k) \Gamma(-k)\Gamma(x+1)}{\Gamma(1+q)\Gamma(-q)\Gamma(x+1-k)}\frac{z_\eta^{q-k}}{(q-k)!}\Big(x-q+2(1-z_\eta)(q-k)\Big).
\end{split}
\label{polynomial form of Q}
\end{equation}
From the above form of $\mathcal{Q}_{q+1}(x;z_\eta)$, one can note that it is a polynomial of degree $q+1$ in the variable $x$ as well because $\frac{\Gamma(x+1)}{\Gamma(x+1-k)}=\begin{cases}1\ \ \ \ \ \ \ \ \ \ \ \ \ \ \ \text{for}\ \  k=0,\\ \prod_{j=0}^{k-1}(x-j)\ \ \text{for}\ \ k\geq 1\end{cases}$. Moreover, the leading power of  $x$ in this polynomial receives contribution from only the $k=q$ term, and its coefficient is $1$. Hence, one can re-express $\mathcal{Q}_{q+1}(x;z_\eta)$ as
\begin{equation}
\begin{split}
\mathcal{Q}_{q+1}(x;z_\eta)
=&\prod_{k=1}^{q+1}\Big(x-\omega_k^{(q+1)}(z_\eta)\Big),
\end{split}
\label{expr of Q in terms of roots}
\end{equation}
where $\omega_k^{(q+1)}(z_\eta)$ are the zeros of $\mathcal{Q}_{q+1}(x;z_\eta)$ when it is treated as a polynomial in $x$.

Now, let us note that from the definition of $\mathcal{Q}_{q+1}(x;z_\eta)$ given in \eqref{definition of Q}, we can  obtain
\begin{equation}
\begin{split}
\mathcal{Q}_{q+1}(n+q;z_\eta)=\frac{\Gamma(n+q+1)}{\Gamma(n)}\Big[& F(1+q,-q,n+1,z_\eta)\\
&-2z_\eta(1-z_\eta)\frac{q(q+1)}{n(n+1)}\ F(2+q,1-q,n+2, z_\eta)\Big].
\end{split}
\end{equation}
Then from the definition of $\Sigma_N(q;z_\eta)$ given in \eqref{defn of SigmaNs}, we have
\begin{equation}
\begin{split}
\Sigma_N(q;z_\eta)
&=\sum_{n=1}^N \ln\Bigg[\frac{\Gamma(n)}{\Gamma(n+q+1)}\mathcal{Q}_{q+1}(n+q;z_\eta)\Bigg].
\end{split}
\end{equation}
Next, using the form of $\mathcal{Q}_{q+1}(x;z_\eta)$ given in \eqref{expr of Q in terms of roots}, we get
\begin{equation}
\begin{split}
\Sigma_N(q;z_\eta)
&=\ln\Bigg[\prod_{k=1}^{q+1}\prod_{n=1}^N \Big(\frac{n+q-\omega_k^{(q+1)}(z_\eta)}{n+k-1}\Big)\Bigg]\\
&=\sum_{k=1}^{q+1}\Bigg[\ln\Gamma\Big(N+1+q- \omega_k^{(q+1)}(z_\eta)\Big)-\ln\Gamma(N+k)\Bigg]\\
&\quad-\sum_{k=1}^{q+1}\ln\Bigg[\frac{\Gamma\Big(1+q- \omega_k^{(q+1)}(z_\eta)\Big)}{\Gamma(k)}\Bigg].
\end{split}
\end{equation}

The asymptotic form of $\Gamma(N+1+a)$ in the large $N$ regime is given by
\begin{equation}
\begin{split}
\ln\Gamma(N+1+a)=N\ln N-N+\Big(a+\frac{1}{2}\Big)\ln N+\frac{1}{2}\ln(2\pi)+O(1/N).
\end{split}
\end{equation}
From this we get the following asymptotic form of $\Sigma_N(q;z_\eta)$ in the large $N$ regime:
\begin{equation}
\begin{split}
\Sigma_N(q;z_\eta)
&=\Big[\frac{1}{2}q(q+1)-\sum_{k=1}^{q+1} \omega_k^{(q+1)}(z_\eta)\Big]\ln N\\
&\quad-\sum_{k=1}^{q+1}\ln\Bigg[\frac{\Gamma\Big(1+q- \omega_k^{(q+1)}(z_\eta)\Big)}{\Gamma(k)}\Bigg]+O(1/N).
\end{split}
\end{equation}

Now, let us determine the quantity $\Big(-\sum_{k=1}^{q+1} \omega_k^{(q+1)}(z_\eta)\Big)$. Note that from the form of $\mathcal{Q}_{q+1}(x;z_\eta)$ given in \eqref{expr of Q in terms of roots}, it is evident that this quantity is the coefficient of $x^{q}$ in the polynomial $\mathcal{Q}_{q+1}(x;z_\eta)$. This coefficient can be extracted from the expression of $\mathcal{Q}_{q+1}(x;z_\eta)$ given in \eqref{polynomial form of Q}, and it turns out to be\footnote{The coefficient of $x^{q}$ in $\mathcal{Q}_{q+1}(x;z_\eta)$ receives contributions from only the $k=q$ and $k=q-1$ terms in the sum given in \eqref{polynomial form of Q}.}
\begin{equation}
\begin{split}
&-\sum_{k=1}^{q+1} \omega_k^{(q+1)}(z_\eta)=-\frac{1}{2}q(q+1)(1+2z_\eta).
\end{split}
\end{equation}
This finally gives the following asymptotic form of $\Sigma_N(q;z_\eta)$,
\begin{equation}
\label{asymp. form of SigmaN}
\begin{split}
\Sigma_N(q;z_\eta)
&=-q(q+1)z_\eta\ln N-\sum_{k=1}^{q+1}\ln\Bigg[\frac{\Gamma\Big(1+q- \omega_k^{(q+1)}(z_\eta)\Big)}{\Gamma(k)}\Bigg]+O(1/N),
\end{split}
\end{equation}
thereby completing the proof of the lemma.
\end{proof}

Substituting the two terms in \eqref{sum over Wn} by their asymptotic forms given in \eqref{asymp. form of first term in sum over Wn} and \eqref{asymp. form of SigmaN}, and then taking the limit $N\rightarrow\infty$, we get
\begin{equation}
\label{infinite sum over W_n}
\begin{split}
\sum_{n=1}^\infty\mathcal{W}_n(q;z_\eta)
&=q(q+1)\Bigg[\gamma z_\eta+(1-z_\eta)\ln(1-z_\eta)\Bigg]-\sum_{k=1}^{q+1}\ln\Bigg[\frac{\Gamma\Big(1+q- \omega_k^{(q+1)}(z_\eta)\Big)}{\Gamma(k)}\Bigg].
\end{split}
\end{equation}

\subsection{Exact expression of the Neumann determinant}

The logarithm of the Neumann determinant of $(\Delta_\eta+m^2)$ for $m^2=-\frac{\eta}{L^2}q(q+1)$ is given by
\be
\begin{split}
\ln {\det}_{\text N} \Big(\Delta_\eta-\frac{\eta}{L^2} q(q+1)\Big)
&=\Big(-\frac{1}{3}-2z_\eta q(q+1)\Big)\ln(L)+\frac{1}{2}\ln(2\pi)+\frac{7}{12}-2\zeta_{\text R}'(-1)\\
&\quad-\frac{2}{3}z_\eta-\frac{1}{3}\ln r_0-2 q(q+1)(\gamma-1+\ln r_0)z_\eta\\
&\quad+\ln \Big[g_0^{(\eta)}(\eta q(q+1))e^{ q(q+1) z_\eta F(1,1,2,z_\eta)}\Big]+2\sum_{n=1}^\infty \mathcal{W}_n(q;z_\eta).
\end{split}
\ee
Now, let us note the identities $z_\eta F(1,1,2,z_\eta)=-\ln(1-z_\eta)$ and 
\be
\begin{split}
g_0^{(\eta)}(\eta q(q+1))
&=-2 q(q+1)z_\eta(1-z_\eta)F(2+q,1-q,2,z_\eta)\\
&=(q+1)\Big[P_{q+1}(1-2z_\eta)-(1-2z_\eta)P_q(1-2z_\eta)\Big],
\end{split}
\ee
where $P_q$ denotes the Legendre polynomial of degree $q$. Using these identities and substituting the series $\sum_{n=1}^\infty\mathcal{W}_n(q;z_\eta)$ by its value obtained in \eqref{infinite sum over W_n}, we finally get the desired exact expression of the Neumann determinant:
\be
\label{Neumann det for special masses}
\begin{split}
&\ln {\det}_{\text N} \Big(\Delta_\eta-\frac{\eta}{L^2} q(q+1)\Big)\\
&=\Big(-\frac{1}{3}-2z_\eta q(q+1)\Big)\ln(L)+\frac{1}{2}\ln(2\pi)+\frac{7}{12}-2\zeta_{\text R}'(-1)\\
&\quad-\frac{2}{3}z_\eta-\frac{1}{3}\ln r_0-2 q(q+1)(-1+\ln r_0)z_\eta\\
&\quad+\ln \Big[(q+1)\Big(P_{q+1}(1-2z_\eta)-(1-2z_\eta)P_q(1-2z_\eta)\Big)\Big]\\
&\quad+q(q+1)(1-2z_\eta)\ln(1-z_\eta)-2\sum_{k=1}^{q+1}\ln\Bigg[\frac{\Gamma\Big(1+q- \omega_k^{(q+1)}(z_\eta)\Big)}{\Gamma(k)}\Bigg].
\end{split}
\ee
Let us remind the reader that $z_\eta$ in the above expression is defined in \eqref{z_eta def.}, and $\omega_k^{(q+1)}(z_\eta)$ for $k\in \{1,\cdots,q+1\}$ are the roots of the polynomial $\mathcal{Q}_{q+1}(x;z_\eta)$ in $x$ that is given in \eqref{polynomial form of Q}.

\subsection{The $q=1$ case}
\label{subsec: q=1 case}
To illustrate the expression given in \eqref{Neumann det for special masses} with an example, let us consider the particular case of $q=1$. In this case the polynomial in  \eqref{polynomial form of Q} reduces to
\begin{equation}
\begin{split}
\mathcal{Q}_{2}(x;z_\eta)=&x^2-(1+2z_\eta)x-2 z_\eta(1-2z_\eta).
\end{split}
\end{equation}
The roots of this polynomial are
 \be
  \omega_1^{(2)}(z_\eta)=\frac{1+2z_\eta+\sqrt{1+12z_\eta-12z_\eta^2}}{2},\   \omega_2^{(2)}(z_\eta)=\frac{1+2z_\eta-\sqrt{1+12z_\eta-12z_\eta^2}}{2}.
 \ee
The Legendre polynomials appearing in \eqref{Neumann det for special masses} are, in this case,
 \be
P_{2}(1-2z_\eta)=\frac{1}{2}\Big[3(1-2z_\eta)^2-1\Big]=1-6z_\eta+6z_\eta^2,\ P_1(1-2z_\eta)=1-2z_\eta.
 \ee
 Using all these, from \eqref{Neumann det for special masses} we get
 \be
\label{Neumann det for q=1}
\begin{split}
{\det}_{\text N} \Big(\Delta_\eta-\frac{2\eta}{L^2}\Big)
&=-\frac{4L^{-\frac{1}{3}-4z_\eta }(2\pi)^{\frac{1}{2}}r_0^{-\frac{1}{3}}z_\eta(1-z_\eta)^{3-4z_\eta}}{\Big[\Gamma\Big(\frac{3-2z_\eta-\sqrt{1+12z_\eta(1-z_\eta)}}{2}\Big)\Gamma\Big(\frac{3-2z_\eta+\sqrt{1+12z_\eta(1-z_\eta)}}{2}\Big)\Big]^2}\\
&\qquad\exp\Big[\frac{7}{12}-2\zeta_{\text R}'(-1)+(\frac{10}{3}-4\ln r_0)z_\eta\Big].
\end{split}
\ee
Let us make a few comments on the above expression of the determinant. Notice that from definition of the quantity $z_\eta$ given in \eqref{z_eta def.} and the ranges of $r_0$ given in \eqref{conformal factors for metrics}, it follows that $z_\eta<1$ for both $\eta=1$ and $\eta=-1$. This means the factor $(1-z_\eta)^{3-4z_\eta}$ in \eqref{Neumann det for q=1} is always a positive quantity. 

For $\eta=-1$, $z_-<0$ and the $\Gamma$ functions in \eqref{Neumann det for q=1} have complex arguments when $z_-<(\frac{1}{2}-\frac{1}{\sqrt{3}})\approx -0.077$. Nonetheless, in this domain the two Gamma functions are complex conjugates of each other and their product is a positive number. For all other allowed values of $z_-$, i.e., $z_-\in[\frac{1}{2}-\frac{1}{\sqrt{3}},0)$, each of the Gamma functions is positive and hence, their product is again positive.  From these observations, it is clear that the determinant in \eqref{Neumann det for q=1} is always positive for $\eta=-1$ as should be expected since  $m^2=\frac{2}{L^2}$ is positive in this case.  

For $\eta=+1$, $z_+>0$ and  $\Gamma\Big(\frac{3-2z_++\sqrt{1+12z_+(1-z_+)}}{2}\Big)$ is always positive. However, the inverse of $\Gamma\Big(\frac{3-2z_+-\sqrt{1+12z_+(1-z_+)}}{2}\Big)$ goes from positive to negative as one increases the value of $z_+$. The transition happens at $z_+=1/2$ where this Gamma function has a pole. This point ($z_+=1/2$) corresponds to $r_0=1$ and the associated positive curvature disk is a hemisphere with its boundary being a geodesic. The determinant on this hemisphere vanishes due to the aforementioned pole of the Gamma function. For all other allowed values of $z_+$, i.e. $z_+\in (0,1)$ and $z_+\neq \frac{1}{2}$, the determinant is strictly negative due to the overall sign in the RHS of \eqref{Neumann det for q=1}.

Let us end with  a few more comments on the vanishing of the determinant on the hemisphere that we noted above. This is actually not only true for the $q=1$ case that we are considering here, but is more generally true for any $q\in\mathbb{N}$. We have proved this in appendix \ref{appendix: hemisphere}. This is not  surprising as the values of $m^2$ that we are considering are negative and the determinant will vanish if $|m^2|$ matches with an eigenvalue of the Laplacian. In appendix \ref{appendix: hemisphere} we show that this is indeed the reason behind the vanishing of the determinant for $m^2=-\frac{q(q+1)}{L^2}$ with $q\in\mathbb{N}$.

\section*{Acknowledgments}

The author would like to thank Frank Ferrari for useful discussions.

\subsection*{Declarations}

\subsubsection*{Funding}

This work is partially supported by the International Solvay Institutes and the Belgian Fonds National de la Recherche Scientifique FNRS (convention IISN 4.4503.15). The author is also supported by a postdoctoral research fellowship of FNRS.

\subsubsection*{Data Availability}

Data sharing not applicable to this article as no datasets were generated or analysed during
the current study.

\subsubsection*{Conflict of interest}

The author has no relevant financial or non-financial interests to disclose.
\appendix

\section{Positivity of  the Laplacian and the Dirichlet-to-Neumann map}
\label{appendix: positivity of operators}
In this appendix we will show that $\Delta$, for both Dirichlet and Neumann boundary conditions, is a positive operator. We will also show that the Dirichlet-to Neumann map $\mathscr{O}_{m^2}$ is  a positive operator as well when $m^2\geq 0$. These results are well-known but we include their proofs for the sake of clarity and completeness. In what follows we will consider these operators for a smooth $d$-dimensional compact manifold $\mathscr{M}$ with a connected boundary $\partial \mathscr{M}$. We will denote the metric on $\mathscr{M}$ by $g$ and the induced metric on  $\partial \mathscr{M}$ by $h$.

To show the positivity of $\Delta$, let us define the following inner product on the space of smooth functions satisfying Dirichlet or  Neumann boundary conditions
\be
\langle\psi_1|\psi_2\rangle=\int_{\mathscr{M}}d^dx \sqrt{g}\psi_1\psi_2.
\ee
Note that for either Dirichlet or Neumann boundary conditions, $\Delta$ is a self-adjoint operator with respect to this inner product as shown below:
\be
\langle\psi_1|\Delta|\psi_2\rangle=\int_{\mathscr{M}}d^dx \sqrt{g}\psi_1\Delta\psi_2=\int_{\mathscr{M}}d^dx \sqrt{g}\psi_2\Delta\psi_1+\oint_{\partial\mathscr{M}}d^{d-1}y\ \sqrt{h}(\psi_2\partial_n\psi_1-\psi_1\partial_n\psi_2).
\ee
The last term in the above equation, involving the boundary integral, vanishes both for Dirichlet and Neumann boundary conditions. So, we get
\be
\langle\psi_1|\Delta|\psi_2\rangle=\int_{\mathscr{M}}d^dx \sqrt{g}\psi_2\Delta\psi_1=\langle\psi_2|\Delta|\psi_1\rangle,
\ee
and hence, $\Delta$ is a self-adjoint operator. Now, the positivity of $\Delta$ can be shown as follows:
\be
\begin{split}
\langle \psi|\Delta|\psi\rangle
&=-\int_{\mathscr{M}}d^dx\ \psi\partial_\mu(\sqrt{g}g^{\mu\nu}\partial_\nu\psi)\\
&=-\oint_{\partial\mathscr{M}}d^{d-1}y\ \sqrt{h} \psi\partial_n\psi+\int_{\mathscr{M}}d^dx\ \sqrt{g}g^{\mu\nu}\partial_\mu\psi \partial_\nu\psi\\
&=\int_{\mathscr{M}}d^dx\ \sqrt{g}g^{\mu\nu}\partial_\mu\psi \partial_\nu\psi\geq 0.
\end{split}
\ee
Here, we have used the fact the boundary integral vanishes for both Dirichlet and Neumann boundary conditions. Moreover, the positivity of the quantity in the last line follows from the positive eigenvalues of the metric.

Similarly, in order to show the positivity of the Dirichlet-to-Neumann map, we can first define the following inner product on the space of smooth functions on the boundary:
\be
\label{inner product: boundary}
\langle f_1| f_2\rangle=\oint_{\partial \mathscr{M}}dy\ \sqrt{h} f_1 f_2.
\ee
To show that $\mathscr{O}_{m^2}$ is a self-adjoint operator with respect to this inner product, consider the following identities:
\be
\label{DN map matrix element}
\begin{split}
\langle f_1|\mathscr{O}_{m^2}| f_2\rangle
&=\oint_{\partial \mathscr{M}}d^{d-1}y\ \sqrt{h} f_1 \partial_n H_{m^2}[f_2]=\int_{ \mathscr{M}}d^dx\ \partial_\mu\Big[H_{m^2}[f_1]\sqrt{g}g^{\mu\nu} \partial_\nu H_{m^2}[f_2]\Big]\\
&=-\int_{ \mathscr{M}}d^dx\ \sqrt{g} H_{m^2}[f_1]\Delta H_{m^2}[f_2]+\int_{ \mathscr{M}}d^dx\ \sqrt{g}g^{\mu\nu} \partial_\mu H_{m^2}[f_1]\partial_\nu H_{m^2}[f_2]\\
&=m^2\int_{ \mathscr{M}}d^dx\ \sqrt{g} H_{m^2}[f_1] H_{m^2}[f_2]+\int_{ \mathscr{M}}d^dx\ \sqrt{g}g^{\mu\nu} \partial_\mu H_{m^2}[f_1]\partial_\nu H_{m^2}[f_2].
\end{split}
\ee
Here we have used the fact that $\Delta H_{m^2}[f_2]=-m^2H_{m^2}[f_2]$ and $H_{m^2}[f_1]|_{\partial \mathscr{M}}=f_1$. Note that the last line of \eqref{DN map matrix element} is symmetric under the exchange of $f_1$ and $f_2$. Therefore, we have
\be
\langle f_1|\mathscr{O}_{m^2}| f_2\rangle=\langle f_2|\mathscr{O}_{m^2}| f_1\rangle,
\ee
and hence, $\mathscr{O}_{m^2}$ is a self-adjoint operator. The positivity of the operator $\mathscr{O}_{m^2}$ for $m^2\geq 0$  also immediately follows from \eqref{DN map matrix element} by setting $f_1=f_2=f$ as shown below:
\be
\begin{split}
\langle f|\mathscr{O}_{m^2}| f\rangle
&=m^2\int_{ \mathscr{M}}d^dx\ \sqrt{g} \Big(H_{m^2}[f]\Big)^2+\int_{ \mathscr{M}}d^dx\ \sqrt{g}g^{\mu\nu} \partial_\mu H_{m^2}[f]\partial_\nu H_{m^2}[f]\geq 0.
\end{split}
\ee

\section{Lowest eigenvalue of the Dirichlet-to-Neumann map in the small $m^2$ regime}
\label{lowest ev of DN map}
In this appendix, we will consider the behaviour of the lowest eigenvalue of  the Dirichlet-to-Neumann map $\mathscr{O}_{m^2}$ for the constant curvature disks in the small $m^2$ regime. As $m^2\rightarrow 0$, the lowest eigenvalue of the Dirichlet-to-Neumann map $\mathscr{O}_{m^2}$ goes to zero. This is because there is a constant mode $f_0^{(0)}(y)=\frac{1}{\sqrt{\ell_\eta}}$, for all $y\in\partial\disk$, corresponding to which $H_0[f_0^{(0)}]=f_0^{(0)}=\frac{1}{\sqrt{\ell_\eta}}$ leading to the vanishing of the action of $\mathscr{O}_0$ on this mode as shown below:
\be
\mathscr{O}_0[f_0^{(0)}]=\partial_n H_0[f_0^{(0)}]=0.
\ee
We will employ perturbation theory to determine the lowest eigenvalue for small $m^2$. We will restrict our attention only to the contribution at $O(m^2)$. Suppose, the  lowest eigenvalue of $\mathscr{O}_{m^2}$ and the corresponding eigenfunction are given by $v_0^{(m^2)}$ and $f_0^{(m^2)}$ respectively. Applying perturbation theory about $m^2=0$ we get
\be
v_0^{(m^2)} f_0^{(0)}=\mathscr{O}_{0} f_0^{(m^2)}+\mathscr{O}_{m^2} f_0^{(0)}+O(m^4).
\ee
This implies that
\be
v_0^{(m^2)}=\langle f_0^{(0)}|\mathscr{O}_{0}| f_0^{(m^2)}\rangle+\langle f_0^{(0)}|\mathscr{O}_{m^2}| f_0^{(0)}\rangle+O(m^4),
\ee
where we are using the inner product defined in \eqref{inner product: boundary}. Since $\mathscr{O}_0$ is a self-adjoint operator with respect to this inner product (as argued in appendix \ref{appendix: positivity of operators}), we have
\be
\langle f_0^{(0)}|\mathscr{O}_{0}| f_0^{(m^2)}\rangle=\langle f_0^{(m^2)}|\mathscr{O}_{0}| f_0^{(0)}\rangle=0.
\ee
Therefore,
\be
\begin{split}
v_0^{(m^2)}
&=\langle f_0^{(0)}|\mathscr{O}_{m^2}| f_0^{(0)}\rangle+O(m^4)=f_0^{(0)}\oint_{\partial \disk} ds_\eta\ \partial_n H_{m^2}[f_0^{(0)}]+O(m^4)\\
&=-f_0^{(0)}\int_{ \disk} d^2 x\ e^{2\sigma_\eta} \Delta_\eta H_{m^2}[f_0^{(0)}]+O(m^4)=m^2f_0^{(0)}\int_{ \disk} d^2 x\ e^{2\sigma_\eta}  H_{0}[f_0^{(0)}]+O(m^4)\\
&=m^2\Big(f_0^{(0)}\Big)^2\int_{ \disk} d^2 x\ e^{2\sigma_\eta}  +O(m^4)=\frac{m^2 A_\eta}{\ell_\eta}  +O(m^4),
\end{split}
\ee
where we have used the fact that 
\be
\Delta_\eta H_{m^2}[f_0^{(0)}]=-m^2 H_{m^2}[f_0^{(0)}]=-m^2 H_{0}[f_0^{(0)}]+O(m^4).
\ee

\section{An alternative derivation of the Neumann determinant}
\label{app: alternative approach to Neumann det}

In this appendix we will evaluate the Neumann determinant ${\det}_{\text N}(\Delta_\eta+m^2)$ via an alternative approach  that  closely follows the analysis done in \cite{ChauFerdet} for similar determinants with Dirichlet boundary conditions.  We will restrict our attention to the determinants on the disks with nonzero constant curvature, i.e. the cases $\eta=\pm 1$, but the determinant on the flat disk ($\eta=0$) can be evaluated analogously. We will show that the values of these determinants derived in this alternative way match with those obtained in section \ref{sec: Determinant for general masses}.

 To evaluate these determinants, we consider the following equation satisfied the eigenfunctions of the Laplacian $\Delta_\eta$,
 \be
 \label{eq. egn fn Laplacian}
 \Delta_\eta\Phi=\la \Phi,
 \ee
 with the Neumann boundary condition
  \be
  \label{Neumann bdy cond. on egn fn of Laplacian}
\partial_n\Phi|_{\partial\disk}=0.
 \ee
 The solutions to equation \eqref{eq. egn fn Laplacian} can be sought by decomposing the field $\Phi$ into its Fourier modes as in \eqref{Fourier decomp.: bulk field}. The equations satisfied by these Fourier modes are 
 \be
 \label{Sturm-Liouville eqn.}
 L_j^{(\eta)}\Phi_j(r)=\la \Phi_j(r),
 \ee
where $L_j^{(\eta)}$ are Sturm-Liouville operators defined as
\be
 \label{Sturm-Liouville operator}
 L_j^{(\eta)}=-e^{-2\sigma_\eta(r)} \Big(\frac{d^2}{dr^2}+\frac{1}{r}\frac{d}{dr}-\frac{j^2}{r^2}\Big)\ \text{for}\ j\in\mathbb{Z}.
\ee
The general solution to equation \eqref{Sturm-Liouville eqn.} that is regular at $r=0$ is the function $f_j^{(\eta)}(r;\la L^2)$ defined in \eqref{f_n def.}. We can then take the  independent real functions satisfying the equation \eqref{eq. egn fn Laplacian} to be
\be
\Phi^{(p)}(\la;r,\theta)=\begin{cases} f_p^{(\eta)}(r;\la L^2)\cos\theta,\ \text{for}\ p\geq 0,\\ f_p^{(\eta)}(r;\la L^2)\sin\theta \ \text{for}\ p<0.\end{cases}
\ee
Now, for each value of $p$, imposing the Neumann boundary condition \eqref{Neumann bdy cond. on egn fn of Laplacian} leads to the following constraint on the values of $\la$:
\be
g_p^{(\eta)}(\la L^2)=0,
\ee
where $g_p^{(\eta)}(z)$ is the radial derivative of $f_p^{(\eta)}(r;z)$ at the boundary ($r=1$) and it is given in \eqref{g_n def.}.
This picks out a discrete set of eigenvalues of the Laplacian, which we denote by $\la_{p,k}^{(\eta;\text N)}$ with $k\in \mathbb{N}_0$. We arrange these eigenvalues in the ascending order, i.e., $\la_{p,0}^{(\eta;\text N)}<\la_{p,1}^{(\eta;\text N)}<\cdots$. With these eigenvalues, one can define a $\zeta$-function
\be
\zeta_p^{(\eta;\text N)}(m^2;s)=\sum_{k=0}^\infty \frac{1}{(\la_{p,k}^{(\eta;\text N)}+m^2)^s}
\ee
for each $p\in \mathbb{Z}$. The determinant of the operator $(L_p^{(\eta)}+m^2)$ is defined in terms of this $\zeta$-function as follows:
\be
\label{zeta regularised SL det}
\ln{\det}_{\text{N}}(L_p^{(\eta)}+m^2)=-\partial_s \zeta_p^{(\eta;\text N)}(m^2;s)|_{s=0}.
\ee
Note that for $p=0$, the lowest eigenvalue  ($\la_{0,0}^{(\eta,\text N)}$) is zero as it corresponds to the constant mode on the disk. So, $\ln{\det}_{\text{N}}(L_0^{(\eta)}+m^2)$, as defined in \eqref{zeta regularised SL det}, will diverge in the $m^2\rightarrow 0$ limit. To regularise this divergence in the massless limit, we can define the following $\zeta$-function  by removing the zero eigenvalue:
\be
{\zeta_0^{(\eta;\text N)}}'(0;s)=\sum_{k=1}^\infty \frac{1}{(\la_{0,k}^{(\eta;\text N)})^s}.
\ee
The $\zeta$-regularised deteminant of $L_0^{(\eta)}$ is given in terms of this $\zeta$-function as folows:
\be
\ln{\det}_{\text{N}}'(L_0^{(\eta)})=-\partial_s {\zeta_0^{(\eta;\text N)}}'(0;s)|_{s=0}.
\ee

Having introduced the determinants of $(L_p^{(\eta)}+m^2)$ both for nonzero mass and zero mass, let us next look at the ratios $\frac{{\det}_{\text{N}}(L_p^{(\eta)}+m^2)}{{\det}_{\text{N}}(L_p^{(\eta)})}$ for $p\neq 0$ and $\frac{{\det}_{\text{N}}(L_0^{(\eta)}+m^2)}{{\det}_{\text{N}}'(L_0^{(\eta)})}$. The logarithms of these ratios are given by
\be
\label{logratiodetSL}
\begin{split}
&\ln\Big[\frac{{\det}_{\text{N}}(L_p^{(\eta)}+m^2)}{{\det}_{\text{N}}(L_p^{(\eta)})}\Big]=\sum_{k\geq 0}\ln\Big(1+\frac{m^2}{\la_{p,k}^{(\eta;\text N)}}\Big)\ \ \text{for}\ \ p\neq 0,\\
&\ln\Big[\frac{{\det}_{\text{N}}(L_0^{(\eta)}+m^2)}{{\det}_{\text{N}}'(L_0^{(\eta)})}\Big]=\ln m^2+\sum_{k> 0}\ln\Big(1+\frac{m^2}{\la_{0,k}^{(\eta;\text N)}}\Big).
\end{split}
\ee
The convergence of the above series can be checked by noting that Weyl's law \cite{Weyllaw1, Weyllaw2} in one dimension gives the following asymptotic form for the eigenvalues $\la_{p,k}^{(\eta;\text N)}$ as $k\rightarrow\infty$:
\be\label{Weyllaw1d}
 \la^{(\eta;\text N)}_{p,k}\underset{k\rightarrow\infty}{\sim} \Bigl(\frac{\pi k}{a_{\eta}}\Bigr)^{2},
\ee
where
\be\label{aetadef} 
a_{\eta} = \int_{0}^{1}\!e^{-\sigma_\eta(r)}\,\d r = 
\begin{cases} 2\arctanh r_{0}  & \text{if $\eta=-1$},\\ 2\arctan r_{0} & \text{if $\eta=+1$} .
\end{cases}\ee
This leads to $\ln\Big(1+\frac{m^2}{\la_{p,k}^{(\eta;\text N)}}\Big)=O(\frac{1}{k^2})$ in the large $k$ regime, which ensures that the series in \eqref{logratiodetSL} are convergent.

Having established the expressions for the ratios of the massive and massless Sturm-Liouville determinants given in \eqref{logratiodetSL},   we will prove in the next subsection that these ratios have certain specific values (see proposition \ref{SLdetcurved}). Then  in subsection \ref{app: subsec: alt deriv. of Neumann det} we will explain how the Neumann determinant ${\det}_\text{N}(\Delta_\eta+m^2)$ can be computed from these ratios.

\subsection{Ratio of the massive and massless Sturm-Liouville determinants}
\label{app:subsec: SL determinant ratio}

\begin{proposition}\label{SLdetcurved}
For the constant curvature round disks with nonzero curvature, i.e. for $\eta=\pm 1$, the ratios of the massive and massless Sturm-Liouville determinants corresponding to Neumann boundary condition are given by
  \be
     \label{SLdetcurved:nonzerop}
  \frac{{\det}_{\text N} (L_p^{(\eta)}+m^2)}{{\det}_{\text N} (L_p^{(\eta)})}=\prod_{k\geq 0}\Big(1+\frac{m^2}{\la_{p,k}^{(\eta;\text N)}}\Big)=\frac{g_p^{(\eta)}(-m^2 L^2)}{|p|}\ \text{for}\ p\neq 0,
\ee
\be
   \label{SLdetcurved:zerop}
  \frac{{\det}_{\text N} (L_0^{(\eta)}+m^2)}{{\det}_{\text N}' (L_0^{(\eta)})}=m^2\prod_{k> 0}\Big(1+\frac{m^2}{\la_{0,k}^{(\eta;\text N)}}\Big)= \frac{1+\eta r_0^2}{2 r_0^2 L^2}g_0^{(\eta)}(-m^2 L^2).
\ee
 \end{proposition}

 \begin{proof}

\mbox{}
 
 \underline{$\bf{p\neq 0}$:}
 
 To prove \eqref{SLdetcurved:nonzerop} for $p\neq 0$, let us first note that the function $g_p^{(\eta)}(z)$ given in \eqref{g_n def.} is an entire function in the complex plane. Now, we can define the following function $h_p^{(\eta)}(z)$ which is also entire:
\be
  e^{h_p^{(\eta)}(z)}=\frac{g_p^{(\eta)}(z)}{\prod_{k\geq 0}\Big(1-z /(\la_{p,k}^{(\eta;\text N)}L^2)\Big)}.
\ee
The derivative of the function $h_p^{(\eta)}(z)$ is given by
\be
  \partial_z h_p^{(\eta)}(z)=\frac{\partial_z g_p^{(\eta)}(z)}{g_p^{(\eta)}(z)}-\sum_{k\geq 0}\frac{1}{z-\la_{p,k}^{(\eta;\text N)}L^2}.
\ee
We will now show that this derivative vanishes everywhere in the complex plane. For this, consider a circular contour $\mathcal{C}_n$ which is centred at the origin of the complex plane and whose radius is $R_n=\frac{L^2}{2}\Big(\la_{p,n-1}^{(\eta;\text N)}+\la_{p,n}^{(\eta;\text N)}\Big)$. Since the eigenvalues $\la_{p,n}^{(\eta;\text N)}$ go to infinity as $n\rightarrow\infty$, the radius $R_n$ also goes to infinity in the same limit. For any point $z\in\mathbb{C}$, we can take $n$ to be sufficiently large so that the circle encloses this point, and then use  Cauchy's residue theorem to get
\be
\frac{1}{2\pi i}\oint_{\mathcal{C}_n}\frac{dz'}{z'-z}\frac{\partial_{z'} g_p^{(\eta)}(z')}{g_p^{(\eta)}(z')}=\frac{\partial_z g_p^{(\eta)}(z)}{g_p^{(\eta)}(z)}-\sum_{k= 0}^{n-1}\frac{1}{z-\la_{p,k}^{(\eta;\text N)}L^2}.
\ee
Therefore,
\be
 \partial_z h_p^{(\eta)}(z)=\lim_{n\rightarrow\infty}\frac{1}{2\pi i}\oint_{\mathcal{C}_n}\frac{dz'}{z'-z}\frac{\partial_{z'} g_p^{(\eta)}(z')}{g_p^{(\eta)}(z')}.
\ee
From the asymptotic behaviour of the hypergeometric functions derived in \cite{Jones}, one can check that as $n\rightarrow\infty$, $\Big|\frac{1}{z'-z}\frac{\partial_{z'} g_p^{(\eta)}(z')}{g_p^{(\eta)}(z')}\Big|\sim O(|z'|^{-\frac{3}{2}})$ along $\mathcal{C}_N$. This means that the above contour integral vanishes as $n\rightarrow\infty$, and hence we find that $h_p^{(\eta)}$ is a constant function in the entire complex plane, i.e.
\be
 \partial_z h_p^{(\eta)}(z)=0\ \text{for all}\ z\in\mathbb{C}.
\ee
Now, as $z\rightarrow0$, we have
\be
\label{zero limit of g_p}
     e^{h_p^{(\eta)}(0)}=\lim_{z\rightarrow 0}g_p^{(\eta)}(z)=|p|.
\ee
Then by using the fact that $h_p^{(\eta)}(z)$ is a constant function in the complex plane, we get
\be
\label{g_p expr on complex plane}
   g_p^{(\eta)}(z).=|p|\prod_{k\geq 0}\Big(1-z /(\la_{p,k}^{(\eta;\text N)}L^2)\Big),
\ee
for any $z\in\mathbb{C}$. Setting $z=-m^2L^2$ ,we have
\be
  g_p^{(\eta)}(-m^2L^2)=|p|\prod_{k\geq 0}\Big(1+\frac{m^2}{\la_{p,k}^{(\eta;\text N)}}\Big).
\ee
This immediately implies that
\be
\frac{{\det}_{\text N} (L_p^{(\eta)}+m^2)}{{\det}_{\text N} (L_p^{(\eta)})}= \prod_{k\geq 0}\Big(1+\frac{m^2}{\la_{p,k}^{(\eta;\text N)}}\Big)= \frac{g_p^{(\eta)}(-m^2 L^2)}{|p|}.
\ee

\mbox{}

\underline{$\bf{p=0}$:}

 To prove \eqref{SLdetcurved:zerop}, let us introduce the entire function $h_0^{(\eta)}(z)$ as follows:
\be
  e^{h_0^{(\eta)}(z)}=\frac{g_0^{(\eta)}(z)}{z\prod_{k> 0}\Big(1-z/(\la_{0,k}^{(\eta;\text N)} L^2)\Big)}.
\ee
The derivative of this function is given by    
\be
 \partial_z h_0^{(\eta)}(z)=\frac{\partial_z g_0^{(\eta)}(z)}{g_0^{(\eta)}(z)}-\frac{1}{z}-\sum_{k> 0}\frac{1}{z-\la_{0,k}^{(\eta;\text N)}L^2}=\frac{\partial_z g_0^{(z)}(z)}{g_0^{(\eta)}(z)}-\sum_{k\geq 0}\frac{1}{z-\la_{0,k}^{(\eta;\text N)}L^2},
\ee
with the second equality following from the fact that $\la_{0,0}^{(\eta;\text N)}=0$. Just as in $p\neq 0$ case, this derivative vanishes everywhere in the complex plane, and hence $h_0^{(\eta)}(z)$ is a constant function in the complex plane. Evaluating the value of this constant function at $z=0$, we get
    \be
  e^{h_0^{(\eta)}(0)}=\lim_{z\rightarrow 0}\frac{g_0^{(\eta)}(z)}{z}=-\frac{2  r_0^2}{(1+\eta r_0^2)^2}F\Big(2,1,2,\frac{\eta r_0^2}{1+\eta r_0^2 }\Big)=-\frac{2  r_0^2}{1+\eta r_0^2}.
\ee
Substituting $e^{h_0^{(\eta)}(z)}$ by this constant value, we have
\be
\label{g_0 expr on complex plane}
 g_0^{(\eta)}(z)=-\frac{2  r_0^2}{1+\eta r_0^2}z\prod_{k> 0}\Big(1-z/(\la_{0,k}^{(\eta;\text N)}L^2)\Big).
\ee
Finally, setting $z=-m^2L^2$, we get
\be
  \frac{{\det}_{\text N} (L_0^{(\eta)}+m^2)}{{\det}_{\text N}' (L_0^{(\eta)})}=m^2\prod_{k> 0}\Big(1+\frac{m^2}{\la_{0,k}^{(\eta;\text N)}}\Big)=\frac{1+\eta r_0^2}{2  r_0^2 L^2}g_0^{(\eta)}(-m^2L^2).
\ee
This completes the proof of Proposition \ref{SLdetcurved}.

\end{proof}

\subsection{Evaluation of the Neumann determinant ${\det}_\text{N}(\Delta_\eta+m^2)$}
\label{app: subsec: alt deriv. of Neumann det}

To relate the Neumann determinant ${\det}_{\text N}(\Delta_\eta+m^2)$ to the Sturm-Liouville determinants, let us first rearrange the all eigenvalues of the Laplacian (i.e. all the $\la_{p,k}^{(\eta;\text N)}$'s) in ascending order as described in section \ref{sec: rel between Neumann and Dirichlet dets}, and label them as $\la_n^{(\eta;\text N)}$ with $n\in\mathbb{N}_0$. Then we can introduce the following Neumann $\zeta$-function:
\be
\zeta^{(\eta,\text N)}(m^2;s)=\sum_{n=0}^\infty \frac{1}{(\la_n^{(\eta;\text N)}+m^2)^s}=\sum_{p\in\mathbb{Z}}\zeta_p^{(\eta,\text N)}(m^2;s).
\ee
The $\zeta$-regularised Neumann determinant of $(\Delta_\eta+m^2)$ is defined by
\be
\ln{\det}_{\text{N}}(\Delta_\eta+m^2)=-\partial_s\zeta^{(\eta,\text N)}(m^2;s)|_{s=0}=-\partial_s\Big[\sum_{p\in\mathbb{Z}}\zeta_p^{(\eta,\text N)}(m^2;s)\Big]\Big|_{s=0}.
\ee
As earlier, in the massless case, one has to remove the zero eigenvalue of $\Delta_\eta$ while defining the zeta function and the determinant. Accordingly, we define
\be
{\zeta^{(\eta,\text N)}}'(0;s)=\sum_{n=1}^\infty \frac{1}{(\la_n^{(\eta;\text N)})^s}={\zeta_0^{(\eta,\text N)}}'(0;s)+\sum_{p\in\mathbb{Z},p\neq 0}\zeta_p^{(\eta,\text N)}(0;s),
\ee
and 
\be
\ln{\det}_{\text{N}}'(\Delta_\eta)=-\partial_s{\zeta^{(\eta,\text N)}}'(0;s)|_{s=0}=-\partial_s\Big[{\zeta_0^{(\eta,\text N)}}'(0;s)+\sum_{p\in\mathbb{Z},p\neq 0}\zeta_p^{(\eta,\text N)}(0;s)\Big]\Big|_{s=0}.
\ee
Note that from the above forms of the massive and massless Neumann determinants, we get
 \be
 \begin{split}
\ln\Big[\frac{{\det}_{\text{N}}(\Delta_\eta+m^2)}{{\det}_{\text{N}}'(\Delta_\eta)}\Big]=-\partial_s\Big[&\zeta_0^{(\eta,\text N)}(m^2;s)-{\zeta_0^{(\eta,\text N)}}'(0;s)\\
&+\sum_{p\in\mathbb{Z},p\neq 0}\Big(\zeta_p^{(\eta,\text N)}(m^2;s)-\zeta_p^{(\eta,\text N)}(0;s)\Big)\Big]|_{s=0}.
\end{split}
\label{reg. logratio det}
\ee
This may suggest a naive expression for $\ln\Big[\frac{{\det}_{\text{N}}(\Delta_\eta+m^2}{{\det}_{\text{N}}'(\Delta_\eta)}\Big]$ of the following form:
\begin{equation}
\begin{split}
\Bigg(\ln\Big[\frac{{\det}_{\text N}(\Delta_\eta+m^2)}{{\det}_{\text N}'(\Delta_\eta)}\Big]\Bigg)_{\text{naive}}
&=\ln\Big[\frac{{\det}_{\text N} (L_0^{(\eta)}+m^2)}{{\det}_{\text N}' (L_0^{(\eta)})}\Big]+\sum_{p\neq 0}\ln\Big[ \frac{{\det}_{\text N} (L_p^{(\eta)}+m^2)}{{\det}_{\text N} (L_p^{(\eta)})}\Big]\\
&=\ln(m^2)+\sum_{k=1}^\infty\ln \Big(1+\frac{m^2}{\lambda_{0,k}^{(\eta;\text N)}}\Big)+\sum_{p\neq0}^\infty\sum_{k=0}^\infty\ln \Big(1+\frac{m^2}{\lambda_{p,k}^{(\eta;\text N)}}\Big)\\
&=\ln(m^2)+\sum_{n=1}^\infty\ln \Big(1+\frac{m^2}{\lambda_{n}^{(\eta;\text N)}}\Big).
\end{split}
\label{naivelogratio}
\end{equation}
However, Weyl's law in two dimensions \cite{Weyllaw1, Weyllaw2} gives the following asymptotic form for $\la^{(\eta;\text{N})}_{n}$:
\be \label{Weyl law 2d}\la^{(\eta;\text{N})}_{n}\underset{n\rightarrow\infty}{\sim}\frac{4\pi n}{A_{\eta}}\, \cvp\ee
where $A_{\eta}$ is the area of the disk.  From this, we find that
\begin{equation}
\begin{split}
&\ln \Big(1+\frac{m^2}{\lambda_{n}^{(\eta;\text N)}}\Big)\underset{n\rightarrow\infty}{\sim}\frac{m^2 A_\eta}{4\pi n},
\end{split}
\end{equation}
which  implies that the series in the last line of \eqref{naivelogratio} diverges. This divergence arose because in going from \eqref{reg. logratio det} and \eqref{naivelogratio} we interchanged the order of taking the limit $s=0$ and the performing the infinite sum over $p$. However, these two operations do not commute. Another way to see this is to note that the  series expansion
\be
\label{zetadifseriesrep}
\zeta^{(\eta;\text N)}(m^2;s)-{\zeta^{(\eta;\text N)}}'(0;s)=\frac{1}{(m^2)^s}+\sum_{n=1}^\infty\Big[\frac{1}{(\la_n^{(\eta;\text N)}+m^2)^s}-\frac{1}{(\la_n^{(\eta;\text N)})^s}\Big]
\ee
is valid only when $\text{Re}(s)> 0$ or $s=0$. This can be verified by looking at the following asymptotic form of the summand which is obtained from the Weyl's law given in \eqref{Weyl law 2d}:
\be
\frac{1}{(\la_n^{(\eta;\text N)}+m^2)^s}-\frac{1}{(\la_n^{(\eta;\text N)})^s}\underset{n\rightarrow\infty}{\sim}-\frac{s m^2(A_{\eta})^{s+1}}{(4\pi)^{s+1} n^{s+1}}.
\ee
Naively it may seem that one can use the above series representation to evaluate the derivative of the LHS of \eqref{zetadifseriesrep} at $s=0$. However, taking such a derivative of the series in  \eqref{zetadifseriesrep} at $s=0$, results in terms which have the asymptotic form $-\frac{m^2 A^{\eta}}{4\pi n}$. This leads to the divergence in \eqref{naivelogratio} that we had noted earlier.

To avoid this problem, let us remove the would-be divergent piece by considering the function $\Big(\zeta^{(\eta;\text N)}(m^2;s)-{\zeta^{(\eta;\text N)}}'(0;s)+sm^2{\zeta^{(\eta;\text N)}}'(0;s+1)\Big)$ which has the following series representation 
\be
\begin{split}
&\zeta^{(\eta;\text N)}(m^2;s)-{\zeta^{(\eta;\text N)}}'(0;s)+sm^2{\zeta^{(\eta;\text N)}}'(0;s+1)\\
&=\frac{1}{(m^2)^s}+\sum_{n=1}^\infty \Big[\frac{1}{(\la_n^{(\eta;\text N)}+m^2)^s}-\frac{1}{(\la_n^{(\eta;\text N)})^s}+\frac{sm^2}{(\la_n ^{(\eta;\text N)})^{s+1}}\Big]\\
&=\frac{1}{(m^2)^s}+\sum_{k=1}^\infty\Big[\frac{1}{(\la_{0,k}^{(\eta;\text N)}+m^2)^s}-\frac{1}{(\la_{0,k}^{(\eta;\text N)})^s}+\frac{sm^2}{(\la_{0,k} ^{(\eta;\text N)})^{s+1}}\Big]\\
&\quad+\sum_{p\in \mathbb{Z},p\neq 0}\sum_{k=0}^\infty \Big[\frac{1}{(\la_{p,k}^{(\eta;\text N)}+m^2)^s}-\frac{1}{(\la_{p,k}^{(\eta;\text N)})^s}+\frac{sm^2}{(\la_{p,k}^{(\eta;\text N)})^{s+1}}\Big]
\end{split}
\label{zetadifregseriesrep}
\ee
that is valid in the domain $\text{Re}(s)>-1$. Now, we can safely calculate the derivative at $s=0$ which gives
\be
\begin{split}
&\ln\Big[\frac{{\det}_{\text{N}}(\Delta_\eta+m^2)}{{\det}_{\text{N}}'(\Delta_\eta)}\Big]
=-\partial_s\Big[\zeta^{(\eta;\text N)}(m^2;s)-{\zeta^{(\eta;\text N)}}'(0;s)\Big]_{|s=0}\\
&= m^2 C^{(\eta;\text N)}-\Big[\partial_s\Big(\zeta_0^{(\eta;\text N)}(m^2;s)-{\zeta_0^{(\eta;\text N)}}'(0;s)\Big)|_{s=0}+m^2{\zeta_0^{(\eta;\text N)}}'(0;1)\Big]\\
&\quad-\sum_{p\in\mathbb{Z},p\neq 0}\Big[\partial_s\Big(\zeta_p^{(\eta;\text N)}(m^2;s)-\zeta_p^{(\eta;\text N)}(0;s)\Big)|_{s=0}+m^2\zeta_p^{(\eta;\text N)}(0;1)\Big]\\
&= m^2 C^{(\eta;\text N)}+\Big[\ln\Big(\frac{{\det}_{\text N}(L_0^{(\eta)}+m^2)}{{\det}_{\text N}'(L_0^{(\eta)})}\Big)-m^2{\zeta_0^{(\eta;\text N)}}'(0;1)\Big]\\
&\quad+\sum_{p\in\mathbb{Z},p\neq 0}\Big[\ln\Big(\frac{{\det}_{\text N}(L_p^{(\eta)}+m^2)}{{\det}_{\text N}(L_p^{(\eta)})}\Big)-m^2\zeta_p^{(\eta;\text N)}(0;1)\Big],
\end{split}
\label{derzetadif}
\ee
where
\begin{equation}
\begin{split}
& C^{(\eta;\text N)}=\partial_s\Big(s{\zeta^{(\eta;\text N)}}'(0;s+1)\Big)_{|s=0}.
\end{split}
\label{CNdef}
\end{equation}
The quantity $C^{(\eta;\text N)}$ gives the finite piece of ${\zeta^{(\eta;\text N)}}'(0;s)$ at $s=1$ after removing the simple pole at this point. Using the Weyl's law \eqref{Weyl law 2d}, one can find that the residue of ${\zeta^{(\eta;\text N)}}'(0;s)$ at this pole is $\frac{A_\eta}{4\pi}$. So, an equivalent expression for $C^{(\eta;\text N)}$ is
\begin{equation}
\begin{split}
& C^{(\eta;\text N)}=\lim_{s\rightarrow1}\Big({\zeta^{(\eta;\text N)}}'(0;s)-\frac{1}{s-1}\frac{A_\eta}{4\pi}\Big).
\end{split}
\label{CNaltdef}
\end{equation}

The equation \eqref{derzetadif} shows how the ratio of the massive and the massless Neumann determinants is obtained from the ratios of the massive and the massless Sturm-Liouville determinants that we had derived in the previous subsection. Substituting these ratios of the massive and the massless Sturm-Liouville determinants by their values given in proposition \ref{SLdetcurved}, we get
\be
\label{inf series rep of Neumann det ratio}
\begin{split}
\ln\Big[\frac{{\det}_{\text{N}}(\Delta_\eta+m^2)}{{\det}_{\text{N}}'(\Delta_\eta)}\Big]
&= m^2 C^{(\eta;\text N)}+\Big[\ln\Big(\frac{1+\eta r_0^2}{2 r_0^2 L^2}g_0^{(\eta)}(-m^2 L^2)\Big)-m^2{\zeta_0^{(\eta;\text N)}}'(0;1)\Big]\\
&\quad+\sum_{p\in\mathbb{Z},p\neq 0}\Big[\ln\Big(\frac{g_p^{(\eta)}(-m^2 L^2)}{|p|}\Big)-m^2\zeta_p^{(\eta;\text N)}(0;1)\Big].
\end{split}
\ee
The only things that remain to be evaluated to completely determine the above infinite series are the quantities ${\zeta_0^{(\eta;\text N)}}'(0;1)$, $\zeta_p^{(\eta;\text N)}(0;1)$ for $p\neq 0$, and $C^{(\eta;\text N)}$. The quantities  ${\zeta_0^{(\eta;\text N)}}'(0;1)$ and $\zeta_p^{(\eta;\text N)}(0;1)$ for $p\neq 0$ can  be evaluated from certain relations that appeared in the proof of proposition \ref{SLdetcurved} as we will explain shortly. The quantity $C^{(\eta;\text N)}$ can be evaluated by considering the Neumann Green's functions of the scalar field in a manner analogous to that explained in \cite{ChauFerdet} for the Dirichlet boundary conditions. However, let us present an alternative representation of $\ln\Big[\frac{{\det}_{\text{N}}(\Delta_\eta+m^2}{{\det}_{\text{N}}'(\Delta_\eta)}\Big]$ that would allow us to bypass the actual computation of  $C^{(\eta;\text N)}$ and corroborate the result obtained in section \ref{sec: Determinant for general masses}. For this, let us note that in case of Dirichlet boundary condition, one can define  $\zeta$-functions from the eigenvalues of the Laplacian and the Sturm-Liouville operators in a manner analogous to those defined for the Neumann boundary conditions in this paper\footnote{These $\zeta$-functions have been discussed in detail in \cite{ChauFerdet}.}. These $\zeta$-functions are as follows:
\be
\begin{split}
&\zeta_p^{(\eta,\text D)}(0;s)=\sum_{k=0}^\infty \frac{1}{(\la_{p,k}^{(\eta;\text D)})^s} \ \ \text{for}\ \  p\in\mathbb{Z},\\
& \zeta^{(\eta,\text D)}(0;s)=\sum_{n=0}^\infty \frac{1}{(\la_{n}^{(\eta;\text D)})^s}=\sum_{p\in\mathbb{Z}}\zeta_p^{(\eta,\text D)}(0;s).
\end{split}
\ee
Here $\la_{p,k}^{(\eta;\text D)}$ are the eigenvalues of the Sturm-Liouville operator $L_p^{(\eta)}$, and $\la_{n}^{(\eta;\text D)}$ are the eigenvalues of $\Delta_\eta$, each of them being now defined for Dirichlet boundary condition. The difference with the Neumann case is that all these eigenvalues are strictly positive because the Dirichlet boundary condition does not allow for a nonzero constant mode on the disk. The main point to note here is that the asymptotic form of the eigenvalues  $\la_{n}^{(\eta;\text D)}$ following from Weyl's law \cite{Weyllaw1, Weyllaw2} is exactly the same as in the Neumann case (given in \eqref{Weyl law 2d}) at leading order, i.e.,
\be 
\label{Weyl law 2d Dirichlet}\la^{(\eta;\text{D})}_{n}\underset{n\rightarrow\infty}{\sim}\frac{4\pi n}{A_{\eta}}.
\ee
This leads to $\zeta^{(\eta,\text D)}(0;s)$ having the same residue as ${\zeta^{(\eta;\text N)}}'(0;s)$ at the pole $s=1$. Consequently, one can define the following quantity analogous to $C^{(\eta;\text N)}$ (see \eqref{CNaltdef}) which gives the finite piece of $\zeta^{(\eta;\text D)}(0;s)$ at $s=1$:
\be
\begin{split}
& C^{(\eta;\text D)}=\lim_{s\rightarrow1}\Big(\zeta^{(\eta;\text D)}(0;s)-\frac{1}{s-1}\frac{A_\eta}{4\pi}\Big).
\end{split}
\label{CDdef}
\ee
From this we can conclude that the difference between $C^{(\eta;\text N)}$ and $C^{(\eta;\text D)}$ should be a finite quantity that is given by
\be
\label{diff CN and CD}
\begin{split}
 C^{(\eta;\text N)}-C^{(\eta;\text D)}
 &=\lim_{s\rightarrow1}\Big({\zeta^{(\eta;\text N)}}'(0;s)-\zeta^{(\eta;\text D)}(0;s)\Big)\\
&={\zeta_0^{(\eta;\text N)}}'(0;1)-\zeta_0^{(\eta;\text D)}(0;1)+\sum_{p\in\mathbb{Z},p\neq 0}\Big[\zeta_p^{(\eta;\text N)}(0;1)-\zeta_p^{(\eta;\text D)}(0;1)\Big].
\end{split}
\ee
To check that the infinite series in the above equation is indeed convergent, let us note that $\zeta_p^{(\eta;\text D)}(0;1)$ was evaluated in \cite{ChauFerdet} and it is given by\footnote{In \cite{ChauFerdet}, $\zeta_p^{(\eta;\text D)}(0;1)$ was computed after setting $L=1$. To obtain the expression in \eqref{zeta_p Dirichlet}, one needs to restore the $L$-dependence by dimensional analysis. Such an analysis gives $\zeta_p^{(\eta;\text D)}(0;1)=L^2 \zeta_p^{(\eta;\text D)}(0;1)|_{L=1}$. See the discussion just below \eqref{CD for L=1}.}
\be
\label{zeta_p Dirichlet}
\begin{split}
\zeta_p^{(\eta;\text D)}(0;1)
&=\frac{A_\eta}{4\pi}\frac{1}{|p|+1}F\Big(1,1,|p|+2,\frac{\eta r_0^2}{1+\eta r_0^2}\Big)\\
&=\frac{L^2  r_0^2}{1+\eta r_0^2}\frac{1}{|p|+1}F\Big(1,1,|p|+2,\frac{\eta r_0^2}{1+\eta r_0^2}\Big).
\end{split}
\ee
To derive the form of $\zeta_p^{(\eta;\text N)}(0;1)$ for $p\neq 0$, note that from \eqref{g_p expr on complex plane}, we have
  \be
 \ln g_p^{(\eta)}(z)=\ln(|p|)+ \sum_{k\geq 0}\ln\Big(1-\frac{z}{\la_{p,k}^{(\eta;\text N)}L^2}\Big).
\ee
for  $z\in\mathbb{C}$. This implies that
  \be
-\partial_z \ln g_p^{(\eta)}(z)|_{z=0}=\sum_{k\geq 0}\frac{1}{\la_{p,k}^{(\eta;\text N)}L^2-z}\Big|_{z=0}=\frac{1}{L^2}\sum_{k\geq 0}\frac{1}{\la_{p,k}^{(\eta;\text N)}}=\frac{\zeta_p^{(\eta;\text N)}(0;1)}{L^2},
\ee
from which we get
  \be
\zeta_p^{(\eta;\text N)}(0;1)=-\frac{L^2}{  g_p^{(\eta)}(z)}\partial_z  g_p^{(\eta)}(z)|_{z=0}=-\frac{L^2}{|p|}\partial_z  g_p^{(\eta)}(z)|_{z=0}.
\ee
Here we have used the fact that $\lim_{z\rightarrow 0}g_p^{(\eta)}(z)=|p|$ as was obtained in \eqref{zero limit of g_p}. Substituting $g_p^{(\eta)}(z)$ by its expression given in \eqref{g_n def.}, we get
\be
\begin{split}
\zeta_p^{(\eta;\text N)}(0;1)
&=L^2\Big[-\eta(\partial_a-\partial_b)F\Big(a,b,|p|+1,\frac{\eta r_0^2 }{1+\eta r_0^2}\Big)\Big|_{a=1,b=0}\\
&\qquad\quad+\frac{2 r_0^2}{|p|(|p|+1)(1+\eta r_0^2)^2}F\Big(2,1,|p|+2,\frac{\eta r_0^2}{1+\eta r_0^2 }\Big)\Big]\\
&=\frac{L^2  r_0^2}{1+\eta r_0^2}\frac{1}{|p|+1}\Big[F\Big(1,1,|p|+2,\frac{\eta r_0^2 }{1+\eta r_0^2}\Big)\\
&\qquad\qquad\qquad\qquad+\frac{2}{|p|(1+\eta r_0^2)}F\Big(2,1,|p|+2,\frac{\eta r_0^2}{1+\eta r_0^2 }\Big)\Big].
\end{split}
\label{zeta_p Neumann}
\ee
Here, in deriving the second equality,  we have used the identities
\be
\partial_aF(a,b,|p|+1,z)|_{a=1,b=0}=0,\ \partial_b F(a,b,|p|+1,z)|_{a=1,b=0}=\frac{z}{|p|+1}F(1,1,|p|+2,z),
\ee
and the fact that $\eta^2=1$ for $\eta=\pm 1$. Although it is not necessary to derive the form of ${\zeta_0^{(\eta;\text N)}}'(0;1)$ for checking the convergence of the series in \eqref{diff CN and CD}, let us mention that it can be derived in exactly the same manner as above starting from the relation \eqref{g_0 expr on complex plane}, and its value is
\be
{\zeta_0^{(\eta;\text N)}}'(0;1)=L^2\Big[-\frac{\ln(1+\eta r_0^2)}{r_0^2}+\eta\Big].
\ee
Now, from the expressions of $\zeta_p^{(\eta;\text D)}(0;1)$ and $\zeta_p^{(\eta;\text N)}(0;1)$ given in \eqref{zeta_p Dirichlet} and \eqref{zeta_p Neumann} respectively, we can see that the terms of the series in \eqref{diff CN and CD} are of the form
 \be
\zeta_p^{(\eta;\text N)}(0;1)-\zeta_p^{(\eta;\text D)}(0;1)
=\frac{2L^2  r_0^2}{(1+\eta r_0^2)^2}\frac{1}{|p|(|p|+1)}F\Big(2,1,|p|+2,\frac{\eta r_0^2}{1+\eta r_0^2 }\Big).
\ee
Using the asymptotic form of the hypergeometric function given in \eqref{hypgeom asymp large c}, we get the following behaviour of these terms in the large $|p|$ regime:
 \be
\zeta_p^{(\eta;\text N)}(0;1)-\zeta_p^{(\eta;\text D)}(0;1)
=\frac{2L^2  r_0^2}{(1+\eta r_0^2)^2}\frac{1}{|p|^2}+O(\frac{1}{|p|^3}).
\ee
From this asymptotic form of the terms of the series in \eqref{diff CN and CD}, we can conclude that this series is convergent.

Substituting $C^{(\eta;\text N)}$ in \eqref{inf series rep of Neumann det ratio} by its expression that can be obtained from \eqref{diff CN and CD}, we find the following infinite series representation of $\ln\Big[\frac{{\det}_{\text{N}}(\Delta_\eta+m^2)}{{\det}_{\text{N}}'(\Delta_\eta)}\Big]$:
\be
\label{alt inf series rep of Neumann det ratio}
\begin{split}
\ln\Big[\frac{{\det}_{\text{N}}(\Delta_\eta+m^2)}{{\det}_{\text{N}}'(\Delta_\eta)}\Big]
&= m^2 C^{(\eta;\text D)}+\Big[\ln\Big(\frac{1+\eta r_0^2}{2 r_0^2 L^2}g_0^{(\eta)}(-m^2 L^2)\Big)-m^2 \zeta_0^{(\eta;\text D)}(0;1)\Big]\\
&\quad+\sum_{p\in\mathbb{Z},p\neq 0}\Big[\ln\Big(\frac{g_p^{(\eta)}(-m^2 L^2)}{|p|}\Big)-m^2\zeta_p^{(\eta;\text D)}(0;1)\Big].
\end{split}
\ee
The value of $C^{(\eta;\text D)}$ was computed in \cite{ChauFerdet} after setting $L=1$, and it is given by
\be
\label{CD for L=1}
C^{(\eta;\text D)}|_{L=1}=\frac{A_\eta|_{L=1}}{2\pi}(\gamma-1+\ln r_0).
\ee
To restore the $L$-dependence in $C^{(\eta;\text D)}$, note that the eigenvalues $\la_n^{(\eta;\text D)}$ scale as $L^{-2}$ with $L$ which implies that $\zeta^{(\eta;\text D)}(0;s)$ scales as $L^{2s}$ with $L$. So, the expansion of $\zeta^{(\eta;\text D)}(0;s)$ around $s=1$ goes as
\be
\begin{split}
\zeta^{(\eta;\text D)}(0;s)
&=L^{2}\Big[1+2(s-1)\ln L+O((s-1)^2)\Big]\\
&\qquad\Big[\frac{1}{s-1}\frac{A_\eta|_{L=1}}{4\pi}+C^{(\eta;\text D)}|_{L=1}+O(s-1)\Big]\\
&=\frac{1}{s-1}\frac{A_\eta}{4\pi}+L^{2} C^{(\eta;\text D)}|_{L=1}+\frac{A_\eta}{2\pi}\ln L+O(s-1),
\end{split}
\ee
where we have used the fact that $A_\eta=L^2 A_\eta|_{L=1}$. From the above expansion of $\zeta^{(\eta;\text D)}(0;s)$ around $s=1$ we can read off $C^{(\eta;\text D)}$, and it is given by
\be
\label{value of CD}
C^{(\eta;\text D)}=L^{2} C^{(\eta;\text D)}|_{L=1}+\frac{A_\eta}{2\pi}\ln L=\frac{A_\eta}{2\pi}(\gamma-1+\ln r_0+\ln L).
\ee
Substituting $C^{(\eta;\text D)}$ and $\zeta_p^{(\eta;\text D)}$ in \eqref{alt inf series rep of Neumann det ratio} by their values given in \eqref{value of CD} and \eqref{zeta_p Dirichlet} respectively, we get
\be
\label{alt inf series rep of Neumann det ratio:final exp}
\begin{split}
\ln\Big[\frac{{\det}_{\text{N}}(\Delta_\eta+m^2)}{{\det}_{\text{N}}'(\Delta_\eta)}\Big]
&=\Big(-2+\frac{A_\eta m^2}{2\pi}\Big)\ln L+\ln\Big(\frac{1+\eta r_0^2}{2 r_0^2 }\Big) +\frac{ A_\eta m^2}{2\pi}(\gamma-1+\ln r_0)\\
&\quad+\ln\Big[g_0^{(\eta)}(-m^2 L^2) e^{-\frac{A_\eta m^2}{4\pi}F(1,1,2,\frac{\eta A_\eta}{4\pi L^2})}\Big]\\
&\quad+\sum_{p\in\mathbb{Z},p\neq 0}\ln\Big[\frac{g_p^{(\eta)}(-m^2 L^2)}{|p|}e^{-\frac{A_\eta m^2}{4\pi(|p|+1)}F(1,1,|p|+2,,\frac{\eta A_\eta}{4\pi L^2})}\Big].
\end{split}
\ee
Now, let us note that the logarithm of the massless Neumann determinant ${\det}_{\text{N}}'(\Delta_\eta)$ can be extracted by combining the equations \eqref{Weyl-invariant Neumann}, \eqref{Liouville action} and \eqref{massless Neumann det: flat unit disk}. It is given by
 \be
 \label{massless Neumann det expr.}
 \begin{split}
 \ln {\det}_{\text N}' (\Delta_\eta)
 &=\Big(2-\frac{1}{3}\Big)\ln L -\frac{1}{2}\ln(2\pi)+\frac{7}{12}-2\zeta_{\text R}'(-1)+\ln\Big(\frac{4\pi r_0^2}{1+\eta r_0^2}\Big)\\
 &\quad-\frac{\eta A_\eta}{6\pi L^2}- \frac{1}{3} \ln(r_0 ).
 \end{split}
 \ee
Combining \eqref{alt inf series rep of Neumann det ratio:final exp} and \eqref{massless Neumann det expr.}, we finally get the following expression for the logarithm of the Neumann determinant of $(\Delta_\eta+m^2)$:
\be
\label{Neumann det gen masses via alt approach}
\begin{split}
\ln\Big[{\det}_{\text{N}}(\Delta_\eta+m^2)\Big]
&=\Big(-\frac{1}{3}+\frac{A_\eta m^2}{2\pi}\Big)\ln L +\frac{1}{2}\ln(2\pi)+\frac{7}{12}-2\zeta_{\text R}'(-1)\\
 &\quad-\frac{\eta A_\eta}{6\pi L^2}- \frac{1}{3} \ln(r_0 )+\frac{ A_\eta m^2}{2\pi}(\gamma-1+\ln r_0)\\
&\quad+\ln\Big[g_0^{(\eta)}(-m^2 L^2) e^{-\frac{A_\eta m^2}{4\pi}F(1,1,2,\frac{\eta A_\eta}{4\pi L^2})}\Big]\\
&\quad+\sum_{p\in\mathbb{Z},p\neq 0}\ln\Big[\frac{g_p^{(\eta)}(-m^2 L^2)}{|p|}e^{-\frac{A_\eta m^2}{4\pi(|p|+1)}F(1,1,|p|+2,,\frac{\eta A_\eta}{4\pi L^2})}\Big].
\end{split}
\ee
Using the fact that $g_p^{(\eta)}(-m^2 L^2)=g_{-p}^{(\eta)}(-m^2 L^2)$, one can clearly see that the above expression is in perfect agreement with the result obtained in \eqref{Neumann determinants:curved disk}.

\section{An example: Neumann determinant on a hemisphere}
\label{appendix: hemisphere}

In this appendix we will discuss a particular case of the positive curvature disk, viz. the case where the disk corresponds to a hemisphere. We will show that the infinite series representation for the Neumann determinant with an arbitrary mass can be exactly computed in this case and the result can be expressed in terms of the Barnes $G$-function. The similar computation for Dirichlet boundary condition was done in \cite{ChauFerdet}. As we have discussed earlier, there would be additional contributions to the logarithm of the Neumann determinant from the integral over the extrinsic curvature  at the boundary and the determinant of the Dirichlet-to-Neumann map. We will compute the contribution of the determinant of the Dirichlet-to-Neumann map from the general result derived in section \ref{sec: det DN map}. This calculation of the determinant of the Dirichlet-to-Neumann map for the hemisphere was also done in \cite{KMW}. We will show that our result matches exactly with theirs. We will then add the contribution of the logarithm of this determinant to the logarithm of the Dirichlet determinant obtained in \cite{ChauFerdet} to extract the result for the logarithm of the Neumann determinant.

Having outlined the strategy, let us  proceed with the computation of the determinant of the Dirichlet-to-Neumann map for the hemisphere. In case of the hemisphere with radius $L$, we have $\eta=1$ and $r_0=1$. Hence, the parameter $\frac{r_0^2}{1+r_0^2}$ that enters into the hypergeometric functions in the definitions of $f_p^{(+)}(1;z)$ and $g_p^{(+)}(z)$ as an argument (see \eqref{f_n def.} and \eqref{g_n def.}) is $\frac{1}{2}$. In this case, the hypergeometric functions simplify due to the general identities given below:
\be
\begin{split}
&F\Big(a,1-a,c,\frac{ 1 }{2}\Big)=\frac{2^{1-c}\sqrt{\pi}\Gamma(c)}{\Gamma(\frac{a+c}{2})\Gamma(\frac{c-a+1}{2})},\\
&F\Big(a,3-a,c,\frac{ 1 }{2}\Big)=\frac{2^{3-c}\sqrt{\pi}\Gamma(c)}{(a-1)(a-2)}\Big[\frac{c-2}{\Gamma(\frac{a+c-2}{2})\Gamma(\frac{c-a+1}{2})}-\frac{2}{\Gamma(\frac{a+c-3}{2})\Gamma(\frac{c-a}{2})}\Big].
\end{split}
\ee
Using these identities, we get that  $f_p^{(+)}(1;z)$ and $g_p^{(+)}(z)$ defined in \eqref{f_n def.} and \eqref{g_n def.} reduce to  the following functions:
\be
\label{f_p hemisphere}
f_p^{(+)}(1;z)=\frac{2^{-p}\sqrt{\pi}\Gamma(p+1)}{\Gamma(\frac{2p+3+\sqrt{1+4z}}{4})\Gamma(\frac{2p+3-\sqrt{1+4z}}{4})}\ \text{for}\ p\geq 0,
\ee
\be
\label{g_p hemisphere, nonzero p}
\begin{split}
g_p^{(+)}(z)
&=p \Big[\frac{2^{-(p-1)}\sqrt{\pi}\Gamma(p)}{\Gamma(\frac{2p+1+\sqrt{1+4z}}{4})\Gamma(\frac{2p+1-\sqrt{1+4z}}{4})}\Big],\ \text{for}\ p>0.
\end{split}
\ee
and
\be
\label{g_p hemisphere, zero p}
g_0^{(+)}(z)=\sqrt{\pi}\Big[\frac{2}{\Gamma(\frac{1+\sqrt{1+4z}}{4})\Gamma(\frac{1-\sqrt{1+4z}}{4})}\Big].
\ee
Notice that from the above expressions, we get
\be
g_p^{(+)}(z)=p f_{p-1}^{(+)}(1;z)\ \text{for}\ p>0.
\ee
Substituting this in the expression of the logarithm of the determinant of the Dirichlet-to-Neumann map given in \eqref{det DN map final result}, we get
 \be
 \begin{split}
\ln {\det}(\mathscr{O}_{m^2})
&=\ln(2\pi)+\ln\Big(\frac{g_0^{(+)}(-m^2 L^2)}{f_0^{(+)}(1;-m^2 L^2)}\Big)+2\sum_{p=1}^\infty\ln\Big(\frac{ f_{p-1}^{(+)}(1;-m^2L^2)}{f_p^{(+)}(1;-m^2 L^2)}\Big)\\
&=\ln(2\pi)+\ln\Big(\frac{g_0^{(+)}(-m^2 L^2)}{f_0^{(+)}(1;-m^2 L^2)}\Big)+2\lim_{N\rightarrow\infty}\sum_{p=1}^N\ln\Big(\frac{ f_{p-1}^{(+)}(1;-m^2L^2)}{f_p^{(+)}(1;-m^2 L^2)}\Big).
\end{split}
\ee
Here, in the second line we have replaced the infinite sum by the limit of a finite sum. We can now use the fact that
\be
\begin{split}
\sum_{p=1}^N\ln\Big(\frac{ f_{p-1}^{(+)}(1;-m^2L^2)}{f_p^{(+)}(1;-m^2 L^2)}\Big)
&=\sum_{p=1}^N\Big[\ln\Big(f_{p-1}^{(+)}(1;-m^2L^2)\Big)-\ln\Big(f_p^{(+)}(1;-m^2 L^2)\Big)\Big]\\
&=\ln\Big(f_{0}^{(+)}(1;-m^2L^2)\Big)-\ln\Big(f_N^{(+)}(1;-m^2 L^2)\Big),
\end{split}
\ee
to obtain
 \be
 \begin{split}
\ln {\det}(\mathscr{O}_{m^2})
&=\ln(2\pi)+\ln\Big(g_0^{(+)}(-m^2 L^2)f_0^{(+)}(1;-m^2 L^2)\Big)\\
&\quad-2\lim_{N\rightarrow\infty}\ln\Big(f_N^{(+)}(1;-m^2 L^2)\Big).
\end{split}
\ee
In the above expression, we can then substitute $f_0^{(+)}(1;-m^2 L^2)$, $f_N^{(+)}(1;-m^2 L^2)$ and $g_0^{(+)}(-m^2 L^2)$ by their values given in \eqref{f_p hemisphere} and \eqref{g_p hemisphere, zero p} to get
 \be
 \begin{split}
\ln {\det}(\mathscr{O}_{m^2})
&=2\ln(2\pi)-\ln\Big[\Gamma\Big(\frac{1+\sqrt{1-4m^2L^2}}{4}\Big)\Gamma\Big(\frac{3+\sqrt{1-4m^2L^2}}{4}\Big)\Big]\\
&\quad-\ln\Big[\Gamma\Big(\frac{1-\sqrt{1-4m^2 L^2}}{4})\Big)\Gamma\Big(\frac{3-\sqrt{1-4m^2L^2}}{4}\Big)\Big]\\
&\quad-2\lim_{N\rightarrow\infty}\Bigg[\frac{1}{2}\ln(\pi)-N\ln(2)+\ln\Gamma(N+1)\\
&\quad\quad-\ln\Gamma\Big(\frac{2N+3+\sqrt{1-4m^2L^2}}{4}\Big)-\ln\Gamma\Big(\frac{2N+3-\sqrt{1-4m^2L^2}}{4}\Big)\Bigg].
\end{split}
\ee
The large $N$ limit in the above expression actually vanishes as can be verified by using the following asymptotic expansion of the $\Gamma$ function:
 \be
\ln \Gamma(z+h)=(z+h-\frac{1}{2})\ln(z)-z+\frac{1}{2}\ln(2\pi) +O(1/z)\ \ \text{as}\ \ z\rightarrow\infty.
\ee
So, we have
 \be
 \begin{split}
\ln {\det}(\mathscr{O}_{m^2})
&=2\ln(2\pi)-\ln\Big[\Gamma\Big(\frac{1+\sqrt{1-4m^2L^2}}{4}\Big)\Gamma\Big(\frac{3+\sqrt{1-4m^2L^2}}{4}\Big)\Big]\\
&\quad-\ln\Big[\Gamma\Big(\frac{1-\sqrt{1-4m^2 L^2}}{4}\Big)\Gamma\Big(\frac{3-\sqrt{1-4m^2L^2}}{4}\Big)\Big].
\end{split}
\ee
Using the Legendre duplication formula,
 \be
\Gamma(z)\Gamma(z+\frac{1}{2})=2^{1-2z}\sqrt{\pi}\Gamma(2z),
\ee
the product of the two $\Gamma$ functions within each of the logarithms in the above expression can be replaced by a single $\Gamma$ function. Then combining these logarithms, we get
 \be
 \begin{split}
\ln {\det}(\mathscr{O}_{m^2})
&=\ln(2\pi)-\ln\Big[\Gamma\Big(\frac{1+\sqrt{1-4m^2L^2}}{2}\Big)\Gamma\Big(\frac{1-\sqrt{1-4m^2 L^2}}{2}\Big)\Big].
\end{split}
\ee
Finally, using the Euler reflection formula,
 \be
\Gamma(z)\Gamma(1-z)=\frac{\pi}{\sin(\pi z)},
\ee
we find
 \be
 \label{DN map det: hemisphere}
 \begin{split}
\ln {\det}(\mathscr{O}_{m^2})
&=\ln(2\pi)-\ln\Bigg[\frac{\pi}{\sin\Big(\pi(\frac{1-\sqrt{1-4m^2 L^2}}{2})\Big)}\Bigg]=\ln\Bigg[2\cos\Big(\frac{\pi}{2}(\sqrt{1-4m^2 L^2})\Big)\Bigg].
\end{split}
\ee
Let us note that this is exactly the expression of the determinant of the Dirichlet-to-Neumann map for the hemisphere that was found in \cite{KMW}.

Having obtained the determinant of the Dirichlet-to-Neumann map, we just need the Dirichlet determinant and the value of $b_0$ to complete the computation of the Neumann determinant as is evident from the relation given in \eqref{DirNeurel:two dim}. We had obtained the value of $b_0$ in \eqref{b0 for const curv disks final expr} for any constant curvature. For the hemisphere, it simply vanishes because the extrinsic curvature at the boundary vanishes. So, we have $b_0=0$. As for the Dirichlet determinant, it was evaluated in \cite{ChauFerdet}, and has the following expression\footnote{In \cite{ChauFerdet} the Dirichlet determinant was calculated after setting $L=1$. Here we have  restored the $L$-dependence.}:
\be
 \label{Dirichlet det: hemisphere}
\begin{split}
\ln {\det}_{\text D} (\Delta_++m^2)
&=\Big(-\frac{1}{3}+m^2 L^2\Big)\ln(L)-\frac{1}{2}\ln(2\pi)+\frac{1}{4}-2\zeta_{\text R}'(-1)\\
&\quad-m^2 L^2-\ln\Big[\frac{1}{\pi}\cos\Big(\frac{\pi}{2}\sqrt{1-4m^2 L^2}\Big)\Big]\\
&\quad+\ln\Big[G\Big(\frac{1}{2}(1+\sqrt{1-4m^2 L^2})\Big)G\Big(\frac{1}{2}(1-\sqrt{1-4m^2 L^2})\Big)\Big],
\end{split}
\ee
where $G(z)$ is the Barnes $G$-function which satisfies $G(z+1)=\Gamma(z)G(z)$. 
Adding the logarithm of the determinant of the Dirichlet-to-Neumann map given in \eqref{DN map det: hemisphere} to the above expression, we immediately get the following expression for the logarithm of the Neumann determinant:
\be
 \label{log Neumann det: hemisphere}
\begin{split}
\ln {\det}_{\text N} (\Delta_++m^2)
&=\Big(-\frac{1}{3}+m^2 L^2\Big)\ln(L)+\frac{1}{2}\ln(2\pi)+\frac{1}{4}-2\zeta_{\text R}'(-1)-m^2 L^2\\
&\quad+\ln\Big[G\Big(\frac{1}{2}(1+\sqrt{1-4m^2 L^2})\Big)G\Big(\frac{1}{2}(1-\sqrt{1-4m^2 L^2})\Big)\Big].
\end{split}
\ee
Taking the exponential of this gives the following expression of the Neumann determinant on the hemisphere:
\be
\begin{split}
 {\det}_{\text N} (\Delta_++m^2)
&=(2\pi)^{\frac{1}{2}} L^{-\frac{1}{3}+m^2 L^2}\exp\Big[\frac{1}{4}-2\zeta_{\text R}'(-1)-m^2 L^2\Big]\\
&\quad G\Big(\frac{1}{2}(1+\sqrt{1-4m^2 L^2})\Big)G\Big(\frac{1}{2}(1-\sqrt{1-4m^2 L^2})\Big).
\end{split}
\ee

Let us make a comment on this determinant when the expression is analytically continued to negative values of $m^2$, and $m^2$ is set equal to $-\frac{q(q+1)}{L^2}$ for $q\in \mathbb{N}$, which is the case we considered in section \ref{sec: Neumann det special masses}. In this case, we have
\be
\begin{split}
 {\det}_{\text N} \Big(\Delta_+-\frac{q(q+1)}{L^2}\Big)
&=(2\pi)^{\frac{1}{2}} L^{-\frac{1}{3}-q(q+1)}\exp\Big[\frac{1}{4}-2\zeta_{\text R}'(-1)+q(q+1)\Big]\\
&\qquad G(1+q)G(-q).
\end{split}
\ee
Now, $G(-q)$ vanishes for all $q\in\mathbb{N}$. This is because we can write
\be
G(-q)=\lim_{x\rightarrow 0}G(x-q)=\lim_{x\rightarrow 0}\Big[\frac{G(x+1)}{\prod_{i=0}^{q}\Gamma(x-i)}\Big],
\ee
and the $\Gamma$ function has poles at negative integers, while $\lim_{x\rightarrow 0}G(x+1)=G(1)=1$. Therefore, we have 
\be
 {\det}_{\text N} \Big(\Delta_+-\frac{q(q+1)}{L^2}\Big)=0\ \text{for all}\ q\in\mathbb{N}.
\ee
We had shown this to be true for $q=1$ in section \ref{sec: Neumann det special masses} but it is more generally valid for all positive integers as shown above.

The vanishing of the above determinants can be understood in terms of the eigenvalues of the Laplacian. To explain this, let us recall from the discussion in appendix \ref{app: alternative approach to Neumann det}  that these eigenvalues are obtained from the eigenvalues of the Sturm-Liouville operators given in \eqref{Sturm-Liouville operator}. For the hemisphere, these Sturm-Liouville operators are  given by
\be
 \label{Sturm-Liouville operator:hemisphere}
 L_p^{(+)}=-\frac{(1+ r^2)^2}{4L^2} \Big(\frac{d^2}{dr^2}+\frac{1}{r}\frac{d}{dr}-\frac{p^2}{r^2}\Big)\ \text{with}\ p\in\mathbb{Z}.
\ee
The eigenvalues of this operator (for Neumann boundary condition) are determined by solving the equation 
\be
g_p^{(+)}(\la L^2)=0,
\ee
with $\la$ being the corresponding eigenvalue of $ L_p^{(+)}$.
 When $p\geq0$,  we can see from the expression of $g_p^{(+)}$ given in \eqref{g_p hemisphere, nonzero p} and \eqref{g_p hemisphere, zero p} that the solutions to the above equation are the same as the solutions to the equation
 \be
 \label{eqn. for eigenvalue of SL operator: hemisphere}
\frac{1}{\Gamma(\frac{2p+1+\sqrt{1+4\la L^2}}{4})\Gamma(\frac{2p+1-\sqrt{1+4\la L^2}}{4})}=0.
\ee
Let us focus on the function $\Gamma(\frac{2p+1-\sqrt{1+4\la L^2}}{4})$. If we set $\la=\frac{q(q+1)}{L^2}$ with $q\in\mathbb{N}$, the inverse of  this function reduces to $1/\Gamma(\frac{p-q}{2})$, which vanishes for $p=q$.\footnote{In fact, $1/\Gamma(\frac{p-q}{2})$ vanishes whenever $(q-p)$ is a non-negative even integer.} From this we can conclude that  for any $q\in\mathbb {N}$, $\la=\frac{q(q+1)}{L^2}$ is an eigenvalue of $L_q^{(+)}$, and therefore it is also an eigenvalue of the Laplacian $\Delta_+$. This leads to the vanishing of the determinant of $(\Delta_+-\frac{q(q+1)}{L^2})$.

\end{document}